\documentclass[journal]{IEEEtran}
\IEEEoverridecommandlockouts

\usepackage{cite}
\usepackage{amsmath,amssymb,amsfonts}
\usepackage{algorithmic}
\usepackage{graphicx}
\usepackage{textcomp}
\usepackage{bm}
\usepackage{mathtools}
\usepackage{amsthm}
\usepackage{url}
\usepackage[sort&compress]{natbib}
\newtheorem{proposition}{Proposition}
\newtheorem{remark}{Remark}

\def\BibTeX{{\rm B\kern-.05em{\sc i\kern-.025em b}\kern-.08em
    T\kern-.1667em\lower.7ex\hbox{E}\kern-.125emX}}

\newcommand{\mb}[1]{\mathbf{#1}}
\newcommand{\mc}[1]{\mathcal{#1}}
\newcommand{\mr}[1]{\mathrm{#1}}
\newcommand{\E}{\mathbb{E}}
\newcommand{\CN}{\mathcal{CN}}
\newcommand{\diag}{\mathrm{diag}}

\begin{document}

\title{Velocity Index Modulation for Movable Antenna Systems}

\author{Yan Zhang, Indrakshi Dey, Shuaishuai Han, 
Ioannis Krikidis, Nicola Marchetti

\thanks{Y.~Zhang, I.~Dey, and N.~Marchetti are with the CONNECT Centre,
Trinity College Dublin, Ireland (e-mail: {zhangy42@tcd.ie}).
S.~Han and I.~Krikidis are with the Department of Electrical and Computer Engineering,
University of Cyprus, Nicosia, Cyprus.
This work was supported in part by the Chinese Scholarship Council. This work was also supported in part by Taighde Eireann - Research Ireland under Grant 13/RC/2077\_P2, by the EU MSCA Project “COALESCE” under Grant Number 101130739, and by US-Ireland R\&D
Partnership Programme RI-SFI-23/US/3924.
}}

\maketitle

\begin{abstract}
In movable antenna (MA) systems, antenna movement induces Doppler frequency
shifts that are conventionally treated as an impairment requiring mitigation.
In this paper, we propose \emph{Velocity Index Modulation for Movable
Antennas} (VIM-MA), which reframes this Doppler effect as an additional
information-bearing degree of freedom. The transmitter selects the antenna
movement velocity from a pre-designed discrete codebook, so that the resulting
Doppler shift conveys extra index bits beyond those carried by the conventional
modulation symbol. Codebook design is formulated as a spectral efficiency
maximization over the velocity spacing $\delta$ and codebook size $N_v$,
subject to an average-information Cram\'{e}r--Rao-type bound
(AIF-CRB) on velocity estimation accuracy, a physical track length constraint,
and a spatial channel decorrelation constraint. A logarithmic change of
variables renders the problem convex and yields a closed-form solution.
We further establish that the peak codebook velocity equals
$D_{\max}/T_s$, and that the decorrelation-limited spacing always lies below
the Rayleigh Doppler resolution, so that VIM-MA is intrinsically a
super-resolution scheme. A covariance-matched detector is derived that requires
neither per-path angle knowledge nor channel state information. Simulation
results show that the decorrelation-limited codebook, which carries five index
bits over a $10\lambda$ aperture, is attainable only with oracle angle
knowledge, whereas the channel-state-free detector is limited to three bits but
reaches that payload approximately $10$~dB earlier than position-domain
indexing charged a realistic pilot budget.
\end{abstract}


\section{Introduction}
\label{sec:intro}

Sixth-generation (6G) wireless networks motivate physical-layer designs that
exploit not only time, frequency and power but also increasingly flexible
spatial degrees of freedom~\cite{tataria6G}. Movable antenna (MA) systems have
attracted attention in this context because repositioning elements within a
finite spatial region samples favourable channel conditions unavailable to
fixed-position antennas~\cite{zhu2022MA,ma2022MIMOcap}: an MA transmitter can
reshape the effective channel response without additional radio-frequency
chains, making antenna movement a low-hardware-cost route to higher spectral
efficiency.

Most existing MA studies focus on the spatial benefit of antenna
repositioning, namely the ability to exploit channel variations over the
available movement region. However, antenna movement also produces a second
physical effect: it induces a Doppler frequency shift on the transmitted
waveform. In classical wireless channels, Doppler is typically associated with
uncontrolled user, transmitter, or scatterer motion and is often treated as a
source of time variation that must be estimated or mitigated~\cite{clarke1968}.
In an MA system, by contrast, part of the Doppler shift is generated by
deliberate transmitter-side motion. This distinction is important because a
controllable Doppler component can be designed, selected, and detected. It
therefore raises the central question of this work: instead of regarding
movement-induced Doppler only as an impairment, can it be converted into an
additional information-bearing degree of freedom?

Index modulation (IM) offers a natural way to answer this question.
Rather than conveying information solely through the value of a modulation
symbol, IM also embeds bits into the index of a selected resource
\cite{basar2016IM}. This idea has appeared in several forms, including antenna
indices in spatial modulation~\cite{mesleh2008SM}, subcarrier indices in
OFDM-IM~\cite{basar2013OFDMIM}, delay-Doppler resource indices in Joint delay-Doppler index modulation (JDDIM) \cite{tek2024JDDIM}, and position-pattern indices in fluid-antenna-aided IM
\cite{zhu2024FAIM}. These schemes show that a properly chosen physical
resource index can provide an additional data stream without necessarily
requiring a higher-order constellation. Motivated by this principle, we propose
\emph{Velocity Index Modulation for Movable Antennas} (VIM-MA), in which the
transmitter maps index bits to a discrete MA velocity codebook. The selected
velocity produces a predictable Doppler signature at the receiver, while the
conventional QAM symbol simultaneously carries the ordinary modulation bits.

FA-IM~\cite{zhu2024FAIM} and position index modulation for fluid
antennas~\cite{pimfas} embed bits in the selected port or position pattern.
The fluid-antenna MIMO extension in~\cite{faimMIMO} uses position-pattern
indexing with a designed codebook and maximum-likelihood (ML) detection. Most
closely related, the MA-IM framework in~\cite{huang2026MAIM}
formulates movable-antenna index modulation by
discretising the movement region into candidate anchors, derives joint ML and
two-stage detectors, establishes error-probability bounds, and concludes that
optimizing channel-domain separation \emph{alone} is insufficient because the
geometry of the joint index--constellation also matters. VIM-MA differs in the
indexed physical quantity: the index is carried by the \emph{velocity} of the
movement and read from the temporal Doppler signature accumulated within a
slot, rather than from an instantaneous spatial channel realization. To assess
it on the same footing as these position-domain schemes we adopt their detector
and error-probability methodology in Section~\ref{sec:detect} and report
effective throughput at matched reliability in Section~\ref{sec:sim}.

The key difference between VIM-MA and existing position-based
indexing schemes is that velocity is not only an index label; it is also a
physical motion variable. In FA-IM~\cite{zhu2024FAIM}, the selected position is
directly the index-bearing channel state. In VIM-MA, however, a selected
velocity has two coupled consequences. During the slot, it determines the
Doppler signature used for velocity-index detection. At the end of the slot, it
also determines the antenna displacement and hence the spatial channel
correlation between adjacent velocity levels. Therefore, the velocity codebook
must satisfy three simultaneous requirements: adjacent Doppler signatures must
be resolvable at the receiver, the resulting antenna displacement must remain
within the finite MA track, and neighboring velocity-induced channel states
must be sufficiently de-correlated. This coupling between Doppler resolvability
and spatial decorrelation is the main technical challenge addressed in this
paper.

We therefore develop an analytical codebook-design framework. We first
characterise the spatial correlation between channels generated by different
velocity levels, giving a decorrelation constraint controlled by a threshold
$\rho_{\max}$, and then derive an AIF-CRB for velocity
estimation under the multipath MA channel model, giving a Doppler-resolvability
constraint. Together with the finite track length these determine the minimum
feasible spacing and the maximum number of velocity levels.

The key contributions are summarized as follows:
\begin{itemize}
  \item We introduce the VIM-MA architecture and derive a MISO signal
  model in which QAM symbols and velocity indices are transmitted
  simultaneously through the same moving antenna array. We further derive the
  spatial velocity correlation and the average-information
  CRB (AIF-CRB) for Doppler/velocity
  estimation, which together form the two physical constraints of the velocity
  codebook design (Section~\ref{sec:system}).
  \item We formulate velocity-codebook design as an SE maximization
  problem over the velocity spacing and codebook size, subject to Doppler
  resolvability, finite track length, and channel-decorrelation constraints. We
  show that the problem admits a closed-form globally optimal solution through
  a convex reformulation, and we identify the active constraint that determines
  the optimal spacing (Section~\ref{sec:convex}).
  \item We show that the velocity is not identifiable from one slot
  without angle knowledge, and resolve this by exploiting the statistical
  scattering model: the movement-induced Doppler spectrum has support set by
  $v_n/\lambda$ alone. This yields a covariance-matched index detector that
  needs no angle information and no CSI, together with a low-complexity
  Doppler-spread alternative, an ambiguity-based union bound on the index
  error probability, and a complexity comparison
  (Section~\ref{sec:detect}).
  \item We establish the mechanical operating envelope. At the
  track-limited optimum the largest codebook velocity is exactly
  $v_{\max}=D_{\max}/T_s$, and the decorrelation-limited spacing always
  satisfies $\delta^\star T_s=d_c\le0.383\lambda$, so the scheme
  \emph{necessarily} operates below the Rayleigh Doppler resolution
  (Section~\ref{sec:feas}).
  \item We provide numerical results that validate the spatial
  correlation model and the AIF-CRB, illustrate the binding-constraint
  behavior, and benchmark the effective throughput at matched index
  reliability of VIM-MA against conventional MA with adaptive
  modulation, Spatial
  Modulation (SM), FA-IM, and joint delay-Doppler index modulation (JDDIM)
  under common physical constraints (Section~\ref{sec:sim}).
\end{itemize}

\textit{Notation:} Boldface lower/upper-case letters denote vectors/matrices.
$(\cdot)^T$ and $(\cdot)^H$ denote the transpose and Hermitian transpose,
respectively. $\mathcal{CN}(0,\sigma^2)$ denotes a circularly symmetric
complex Gaussian distribution. $J_0(\cdot)$ is the zeroth-order Bessel
function of the first kind. $\lfloor\cdot\rfloor$ denotes the floor function.
The normalized sinc function is
$\mathrm{sinc}(x)=\sin(\pi x)/(\pi x)$, and $Q(\cdot)$ denotes the standard
Gaussian tail probability. Throughout this paper the maximum track
length is denoted $D_{\max}$ rather than $L_{\max}$, to avoid a collision with
the number of propagation paths $L$; likewise the Kronecker delta is written
$\delta_{ll'}^{\mr{K}}$ to avoid a collision with the velocity spacing
$\delta$.

\section{System Model}
\label{sec:system}

\begin{figure*}[t]
\centering
\includegraphics[width=0.70\textwidth]{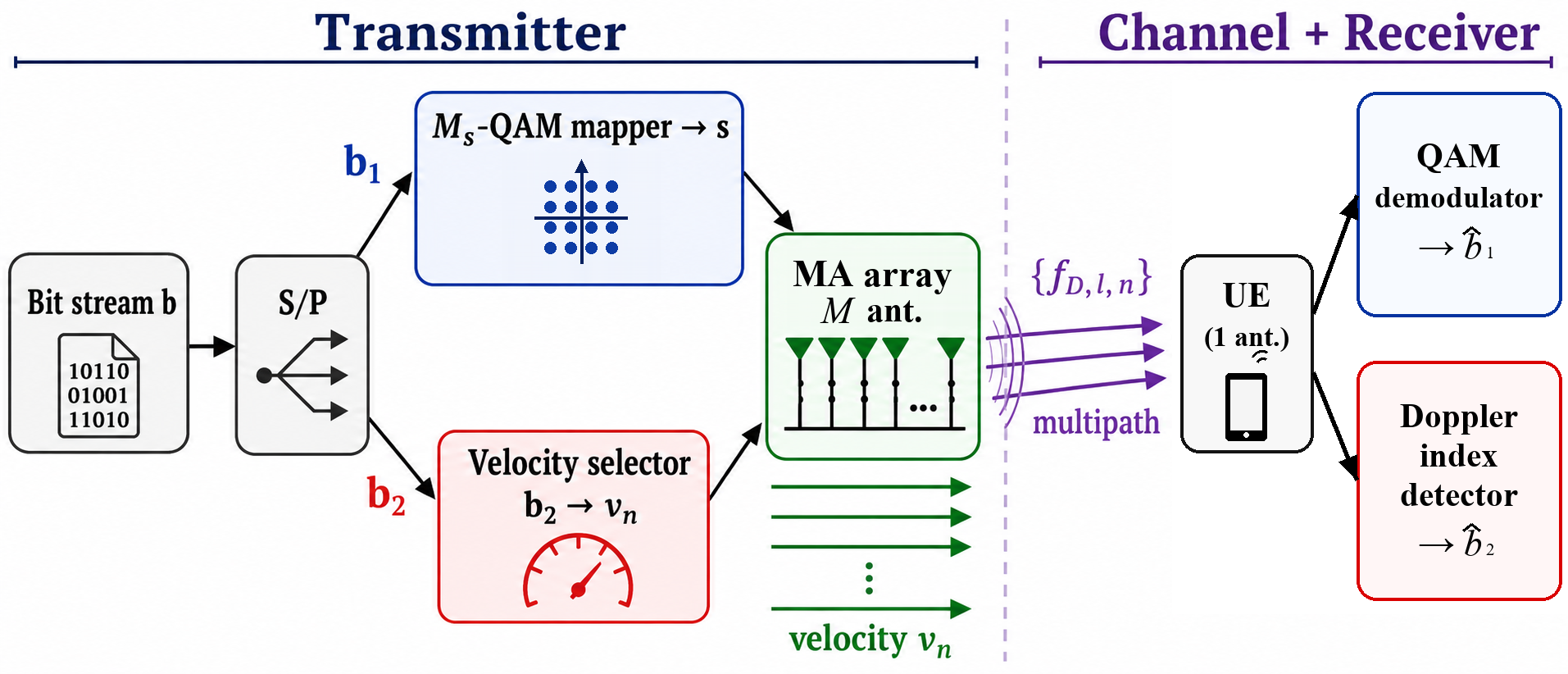}
\caption{VIM-MA architecture. The bit stream $\mathbf{b}$ is split into a
QAM sub-stream $\mathbf{b}_1$ and an index sub-stream $\mathbf{b}_2$. The
index bits select a movement velocity $v_n$ from the codebook $\mathcal{V}$;
all $M$ movable antennas then slide along their tracks at $v_n$, imprinting
a velocity-dependent Doppler signature $\{f_{\mathrm{D},l,n}\}$ on the
multipath channel. The receiver detects the velocity index from the Doppler
content to recover $\hat{\mathbf{b}}_2$ and demodulates the QAM symbol to
recover $\hat{\mathbf{b}}_1$.}
\label{fig:arch}
\end{figure*}

\subsection{MA Channel Model with Movement-Induced Doppler}
\label{sec:channel}

We consider a single-user downlink system in which a base station (BS),
equipped with $M$ movable antennas (MAs) each deployed on a
one-dimensional (1D) track of maximum length $D_{\max}$, transmits to
a single-antenna user equipment (UE). The system adopts a
multiple-input single-output (MISO) configuration with one served user. Fig.~\ref{fig:arch} shows the complete VIM-MA transceiver flow: the
transmitter maps one bit stream to the QAM symbol and the other to the MA
velocity, and the receiver separately demodulates the QAM symbol and detects
the velocity index from the Doppler content of the received samples.

The wireless channel comprises $L$
distinguishable propagation paths. For the velocity-bound
derivation, these paths are assumed resolvable into orthogonal delay bins with
independent post-processing noise, and their AoDs are available at the
receiver from channel estimation; instantaneous path parameters are not
required at the transmitter for codebook design. All $M$ MAs move at a common constant velocity $v$ during each slot of
duration $T_s$, so the position of the $m$-th MA at time $t$ is
$p_m(t) = p_{0,m} + vt$, where $p_{0,m}$ is its initial position.

According to the field-response channel model~\cite{zhu2022MA,zhu2024FAIM},
the channel coefficient of the $m$-th MA at time $t$ is
\begin{equation}
    h_m(t)
    =
    \frac{1}{\sqrt{L}}
    \sum_{l=1}^{L}
    \alpha_l\,
    e^{j\frac{2\pi}{\lambda}p_m(t)\cos\theta_l},
    \label{eq:hm_t}
\end{equation}
where $\alpha_l \sim \mathcal{CN}(0,1)$ and $\theta_l \in (0,\pi)$ denote
the complex path gain and angle of departure (AoD) of the $l$-th path,
respectively, and $\lambda$ is the carrier wavelength. The AoDs
$\{\theta_l\}$ are modeled as i.i.d.\ random variables following the standard isotropic
scattering assumption~\cite{clarke1968} as $\theta_l\sim\mathcal{U}(0,\pi)$,
since the scattering geometry is unknown at the transmitter. Substituting
$p_m(t) = p_{0,m} + vt$ into~\eqref{eq:hm_t} yields
\begin{equation}
    h_m(t)
    =
    \frac{1}{\sqrt{L}}
    \sum_{l=1}^{L}
    \alpha_l\,
    e^{j\frac{2\pi}{\lambda}p_{0,m}\cos\theta_l}\,
    e^{j2\pi f_{{\rm D},l}\,t},
    \label{eq:hm_t_expanded}
\end{equation}
where the movement-induced Doppler shift on path $l$ is
\begin{equation}
    f_{{\rm D},l} \triangleq \frac{v\cos\theta_l}{\lambda}.
    \label{eq:doppler_def}
\end{equation}
Since all $M$ MAs share the same velocity $v$, the Doppler shift $f_{{\rm D},l}$
is identical across all antennas and depends only on the common velocity $v$
and the AoD $\theta_l$.
The overall channel vector is
$\mathbf{h}(t) = [h_1(t),\ldots,h_M(t)]^T \in \mathbb{C}^{M}$.
The received signal at the single-antenna receiver is
\begin{equation}
    y(t)=\mathbf{h}^H(t)\,\mathbf{w}\,s + n(t),
    \label{eq:received_signal}
\end{equation}
where $\mathbf{w}\in\mathbb{C}^{M}$, $\|\mathbf{w}\|=1$, is the transmit
beamforming vector, $s$ is the transmitted symbol with
$\mathbb{E}[|s|^2] = P$, and $n(t)\sim\mathcal{CN}(0,N_0)$ is additive
white Gaussian noise. The normalization $\|\mathbf{w}\|=1$ decouples the beamforming direction from
the transmit power; the specific design (MRT, ZF, \dots) is a system-level
choice that affects the velocity codebook only through the average SNR
in~\eqref{eq:SNR_def}.

The effective path coefficient $\beta_l$ represents the \emph{beamformed
effective gain} of the $l$-th propagation path after coherent combining
across all $M$ antennas. Substituting $\mathbf{h}(t)$
into~\eqref{eq:received_signal} yields
\begin{equation}
  y(t)
  = s\sum_{l=1}^{L}
    \frac{\alpha_l}{\sqrt{L}}
    \!\left(\sum_{m=1}^{M}
    w_m\,e^{j\frac{2\pi}{\lambda}p_{0,m}\cos\theta_l}\right)
    e^{j2\pi f_{{\rm D},l}t}
    +n(t),
  \label{eq:y_expanded}
\end{equation}
where the inner summation over $m$ arises from the beamforming operation.
Defining
\begin{equation}
  \beta_l
  \triangleq
  \frac{\alpha_l}{\sqrt{L}}
  \sum_{m=1}^{M}
  w_m\,e^{j\frac{2\pi}{\lambda}p_{0,m}\cos\theta_l}
  \in\mathbb{C},
  \label{eq:beta}
\end{equation}
the received signal is written compactly as
\begin{equation}
  y(t) = s\sum_{l=1}^{L}\beta_l\,e^{j2\pi f_{{\rm D},l} t} + n(t).
  \label{eq:y_exact}
\end{equation}
The velocity $v$ is selected from a pre-designed codebook, introduced in
Section~\ref{sec:sigmodel}, where it is denoted $v_n$ to indicate the
chosen index. The Doppler shift $f_{{\rm D},l}$ is then written as
$f_{{\rm D},l,n} = v_n\cos\theta_l/\lambda$ throughout the remainder of
this paper.

\subsection{VIM-MA Signal Model and Velocity Codebook}
\label{sec:sigmodel}

\subsubsection{Bit Splitting and Velocity Codebook Lookup}

In each transmission slot the incoming bit stream $\mb{b}$ is split by a
serial-to-parallel (S/P) converter into two sub-streams. The first sub-stream
$\mb{b}_1$ contains $\log_2 M_s$ bits and is mapped onto an $M_s$-ary QAM
symbol $s\in\mathcal{S}$ with $\mathbb{E}[|s|^2]=P$. The second sub-stream
$\mb{b}_2$ contains $\lfloor\log_2 N_v\rfloor$ bits and selects a velocity
$v_n$ from the pre-designed codebook
\begin{equation}
  \mathcal{V} = \{v_1,v_2,\ldots,v_{N_v}\},
  \quad v_n = (n-1)\delta,\quad n=1,\ldots,N_v,
  \label{eq:codebook}
\end{equation}
via a lookup table, where $\delta > 0$ is the
uniform velocity spacing. The codebook is fully characterized by two design
parameters: the number of velocity levels $N_v$ and the spacing $\delta$.
All $M$ MAs move simultaneously at the selected velocity $v_n$ throughout
the slot, so the generic velocity $v$ of Section~\ref{sec:channel} is
identified as $v = v_n$, and the Doppler shift of~\eqref{eq:doppler_def}
becomes
\begin{equation}
  f_{{\rm D},l,n} \triangleq \frac{v_n\cos\theta_l}{\lambda}.
  \label{eq:doppler_n}
\end{equation}
The received signal~\eqref{eq:y_exact} is therefore a superposition of $L$
complex exponentials at frequencies $\{f_{{\rm D},l,n}\}_{l=1}^{L}$.
Although each path experiences a distinct Doppler shift due to its AoD
$\theta_l$, all path Dopplers are jointly determined by the common velocity
$v_n$, which enables velocity-index detection from the overall Doppler
structure. Thus, the spectral efficiency (SE) is
\begin{equation}
  R = \log_2 M_s + \left\lfloor\log_2 N_v\right\rfloor \quad\text{[bpcu]},
  \label{eq:SE}
\end{equation}
where bpcu stands for bits per channel use.

Two properties of~\eqref{eq:SE} are worth making explicit. \emph{(i) Floor operation} : The bit-splitting rule delivers
$\lfloor\log_2 N_v\rfloor$ index bits. Maximising the unfloored quantity would
overstate the rate by up to one bit and would hide a free design margin: once
$N_v^\star$ exceeds a power of two, the \emph{largest} spacing that still
supports $B=\lfloor\log_2 N_v^\star\rfloor$ bits is
\begin{equation}
  \delta_{\mathrm{rel}}
  =\frac{D_{\max}}{(2^{B}-1)\,T_s}\;\ge\;\delta^\star ,
  \label{eq:delta_relaxed}
\end{equation}
which delivers exactly the same payload with a strictly larger separation and
hence a strictly lower index-error probability. Design should therefore report
$(\delta_{\mathrm{rel}},2^{B})$ rather than $(\delta^\star,N_v^\star)$.

\emph{(ii) Rate versus throughput} : Equation~\eqref{eq:SE} counts
\emph{transmitted} bits and is blind to detection errors, so it cannot by
itself substantiate a claim of superiority. We therefore report the effective
throughput
\begin{equation}
  R_{\mathrm{eff}}
  =\left(1-P_{\mathrm{e}}^{\mathrm{idx}}\right)\!\left\lfloor\log_2 N_v\right\rfloor
  +\left(1-P_{\mathrm{e}}^{\mathrm{sym}}\right)\log_2 M_s ,
  \label{eq:Reff}
\end{equation}
and benchmark all schemes at a matched target error rate in
Section~\ref{sec:sim}. An index error also affects the symbol bits: since the QAM
symbol is demodulated conditionally on $\hat n$, a velocity-index error
generally corrupts the symbol decision as well, so~\eqref{eq:Reff} is itself
optimistic unless the second term is conditioned on $\hat n = n$.

\subsubsection{Received Signal Model}

With $v=v_n$ identified from the codebook~\eqref{eq:codebook}, the
received signal~\eqref{eq:y_exact} becomes
\begin{equation}
  y(t) = s\sum_{l=1}^{L}\beta_l\,e^{j2\pi f_{{\rm D},l,n} t} + n(t),
  \label{eq:y_n}
\end{equation}
where $f_{{\rm D},l,n}$ is defined in~\eqref{eq:doppler_n}.
Different velocity indices produce distinct sets of Doppler frequencies,
enabling the receiver to decode $\mb{b}_2$ from the Doppler content of
the received signal.

The receiver takes $N$ uniform samples over the slot at $t_k=kT_s/N$,
$k=0,\ldots,N-1$, stacked into $\mb{y}\in\mathbb{C}^N$. Each path then
contributes a complex exponential sampled at $N$ equally spaced points, so it
is convenient to define the \emph{Doppler steering vector}
\begin{equation}
  \mb{a}(f) \triangleq \left[1,\,e^{j2\pi f\frac{T_s}{N}},\,\ldots,\,
  e^{j2\pi f\frac{(N-1)T_s}{N}}\right]^T \in\mathbb{C}^N,
  \label{eq:steering}
\end{equation}
so that sampling~\eqref{eq:y_n} and stacking yields
\begin{equation}
  \mb{y} = s\sum_{l=1}^{L}\beta_l\,\mb{a}(f_{{\rm D},l,n}) + \mb{n},
  \label{eq:rx}
\end{equation}
where $\mb{n}\sim\mathcal{CN}(\mb{0},N_0\mb{I}_N)$.
Different velocity indices produce different Doppler frequencies and hence
different steering vectors, making the choices distinguishable in $\mb{y}$.
Since $\|\mb{a}(f)\|^2=N$ and cross-path terms vanish in expectation over
i.i.d.\ $\alpha_l$, the signal power per sample is
$P\sum_l\E[|\beta_l|^2]$ with $P=\E[|s|^2]$, against a noise power $N_0$. The
average received SNR per sample is therefore
\begin{equation}
  \mathrm{SNR} \triangleq
  \frac{P \sum_{l=1}^{L}\mathbb{E}_{\alpha_l,\theta_l}[|\beta_l|^2]}{N_0},
  \label{eq:SNR_def}
\end{equation}
where the joint expectation is taken over the random path gains $\alpha_l$ and
AoDs $\theta_l$; the resulting $\mathrm{SNR}$ is a deterministic
system parameter derived in Appendix~\ref{app:crlb}.

Since $\mathrm{SNR}$ is defined \emph{per sample}
while the $N$ samples span a fixed slot, increasing $N$ increases the total
collected energy in proportion, which is why the bound of
Section~\ref{sec:crlb} decays as $N/(N^2-1)\approx1/N$. Were the waveform
oversampled beyond its Nyquist rate the extra samples would carry no new
information. We therefore treat $N$ as the number of \emph{independent}
observations, $N\le\lceil WT_s\rceil$ with $W$ the receiver noise bandwidth,
and fix $W=64$~kHz so that the baseline $N=64$ sits exactly at Nyquist. Any
``temporal processing gain'' below is an energy-accumulation gain.

The receiver estimates the velocity index $n$ from $\mb{y}$ and recovers
$\hat{\mb{b}}_2$; it then demodulates the QAM symbol to recover
$\hat{\mb{b}}_1$.

\subsection{CRLB on Velocity Estimation}
\label{sec:crlb}

We assume that conventional QAM symbol detection is reliable and focus on
the design and detection of velocity indices. The $L$ distinguishable paths
are resolved into orthogonal observation branches, and each $\theta_l$ is
treated as known at the receiver (or, equivalently, the following is an
oracle bound conditioned on the AoD). Without this assumption, the product
$v_n\cos\theta_l$ does not identify $v_n$ separately from an unknown
$\theta_l$ in a single path.

This assumption is essential, because on its own it makes the bound
unattainable. The non-identifiability is not confined to a single path: for
\emph{any} $v'>v$ the substitution
$\theta_l'=\arccos\!\big((v/v')\cos\theta_l\big)$, well defined for every $l$,
reproduces the entire Doppler set $\{v\cos\theta_l/\lambda\}_{l=1}^{L}$
exactly. With unconstrained AoDs the velocity is therefore unidentifiable from
one slot \emph{regardless of $L$ or of the SNR}, and~\eqref{eq:CRLB} is an
oracle bound rather than an achievable one. Identifiability is restored by the
statistical scattering model: under isotropic scattering the movement-induced
Doppler spectrum is supported on $[-v_n/\lambda,v_n/\lambda]$, so $v_n$ is
identifiable from the \emph{spread} of that spectrum without per-path angle
knowledge. We adopt this as the operating assumption of the paper and develop
it in Section~\ref{sec:detect}; \eqref{eq:CRLB} is retained throughout as an
\emph{oracle} benchmark that upper-bounds any AoD-free receiver.

For each propagation path $l$, we construct a $2\times2$ FIM for the
parameter vector $\bm{\xi}_l=[v_n,\,\phi_l]^T$, where
$\phi_l=\angle(\beta_ls)$ is an unknown composite phase and is eliminated
by the Schur complement. The QAM symbol is assumed detected; averaging its
conditional Fisher information over the constellation replaces $|s|^2$ by
$P=\mathbb{E}[|s|^2]$. Thus, define the symbol-averaged path power
$A_l^2\triangleq P|\beta_l|^2$. The conditional per-path CRLB derived in
Appendix~\ref{app:crlb} is
\begin{equation}
  \mathrm{CRLB}_l(v_n\mid\theta_l,\alpha_l)
  = \frac{6\lambda^2 N}
         {(2\pi\cos\theta_l\,T_s)^2(N^2-1)\cdot(A_l^2/N_0)}.
  \label{eq:crlb_cond2}
\end{equation}
For nonconstant-envelope QAM, the corresponding instantaneous expression is
obtained by replacing $P$ with the detected value $|s|^2$.

Because $|\beta_l|^2$ depends on $\theta_l$ through the beamformed spatial
response, its angular dependence must be retained. Define the
beamforming-weighted angular factor
\[
  \eta_{\mb{w}}
  \triangleq
  \frac{\sum_{l=1}^{L}\mathbb{E}[|\beta_l|^2\cos^2\theta_l]}
       {\sum_{l=1}^{L}\mathbb{E}[|\beta_l|^2]}.
\]
After summing the effective information of the orthogonally resolved paths,
averaging it over $\alpha_l$ and $\theta_l$, and then inverting, we obtain the
average-information CRB (AIF-CRB)
\begin{equation}
  \begin{aligned}
  \mathrm{CRB}_{\mathrm{AIF}}(v_n)
  &= \frac{6\lambda^2 N}
  {(2\pi T_s)^2(N^2-1)\cdot\mathrm{SNR}\cdot\eta_{\mb{w}}}\\
  &\approx
  \frac{12\lambda^2 N}
  {(2\pi T_s)^2(N^2-1)\cdot\mathrm{SNR}},
  \end{aligned}
  \label{eq:CRLB}
\end{equation}
where $\mathrm{SNR}$ is defined in~\eqref{eq:SNR_def}. The approximation
uses $\eta_{\mb{w}}\approx1/2$, which Appendix~\ref{app:crlb} establishes
under the large-separation spatial-response approximation. Importantly,
\eqref{eq:CRLB} is the inverse of expected effective information, not
$\mathbb{E}[\mathrm{CRLB}_l]$; it is used here as a statistical codebook-design
surrogate.

The associated minimum admissible velocity spacing is
\begin{equation}
  \begin{aligned}
  \delta_{\min}
  &\triangleq \gamma\sqrt{\mathrm{CRB}_{\mathrm{AIF}}(v_n)}\\
  &=\gamma\lambda\sqrt{\frac{6N}
  {(2\pi T_s)^2(N^2-1)\cdot\mathrm{SNR}\cdot\eta_{\mb{w}}}}\\
  &\approx
  \gamma\lambda\sqrt{\frac{12N}
  {(2\pi T_s)^2(N^2-1)\cdot\mathrm{SNR}}}.
  \end{aligned}
  \label{eq:dvmin}
\end{equation}
The offline codebook uses only the statistical factor $\eta_{\mb{w}}$ and
therefore does not require instantaneous individual AoDs or path gains at the
transmitter, although the conditional per-path derivation assumes the AoDs
are available at the receiver.

The AIF-CRB surrogate addresses estimation-domain resolvability: whether the receiver
can reliably distinguish $v_n$ from adjacent levels given the observation
noise. A second, independent requirement is \emph{channel-domain
distinguishability}: distinct velocity choices must produce sufficiently
different channel realizations.

The operational basis of this second requirement must be stated carefully. It
is \emph{not} that correlated end-of-slot channels prevent index detection:
the index is decoded from the Doppler content \emph{within} the slot
via~\eqref{eq:rx}, so distinct Doppler signatures give distinct received
signals however correlated the channels they leave behind at $t=T_s$. Rather,
\eqref{eq:C3} serves two purposes. First, the end-of-slot displacement sets the
channel seen by the \emph{next} slot, and decorrelation across codebook entries
bounds how much the post-slot channel depends on the transmitted index, which
is what allows a channel estimate to be reused. Second, following
\cite{huang2026MAIM}, index states mapping to nearly identical channel
realizations collapse the joint index--constellation and raise the error
probability even when the index is estimable; decorrelation is then a proxy for
minimum distance in the joint signal space. Remark~\ref{rem:c3inactive} shows
that, once the ambiguity constraint~\eqref{eq:C1b} is included, \eqref{eq:C3}
is not the binding constraint for a realizable receiver.
To characterize this, we evaluate the channel at the end of
the slot. After a slot of duration $T_s$, the $m$-th MA has displaced from
$p_{0,m}$ to $p_{0,m}+vT_s$, giving the end-of-slot channel coefficient
\begin{equation}
  h_m(T_s;\,v) = \frac{1}{\sqrt{L}}\sum_{l=1}^{L}\alpha_l\,
  e^{j\frac{2\pi}{\lambda}p_{0,m}\cos\theta_l}
  e^{j\frac{2\pi}{\lambda}v T_s\cos\theta_l},
  \label{eq:hm_end}
\end{equation}
which is the channel coefficient at the displaced position
$p_{0,m}+vT_s$. The \emph{spatial velocity correlation} between the
channel responses induced by two velocity values $v$ and $v'$ is defined
as $R_h(v,v')\triangleq\mathbb{E}[h_m(T_s;v)\,h_m^*(T_s;v')]$.

\begin{proposition}[Spatial velocity correlation]
\label{prop:corr}
Under the field-response model~\eqref{eq:hm_t} with
$\alpha_l\sim\mathcal{CN}(0,1)$ i.i.d.\ and AoDs
i.i.d.\ $\mathcal{U}(0,\pi)$, the spatial velocity correlation averaged
over both the path gains and the AoDs is, for every finite $L$, exactly
\begin{equation}
  R_h(v,v') = J_0\!\left(\frac{2\pi(v-v')T_s}{\lambda}\right),
  \label{eq:Rh_J0}
\end{equation}
which depends only on the velocity difference $v-v'$ and is independent
of the initial positions $\{p_{0,m}\}$. If instead the AoD realization
$\{\theta_l\}$ is held fixed and only the gains are averaged, the empirical
correlation $\frac{1}{L}\sum_l e^{j2\pi(v-v')T_s\cos\theta_l/\lambda}$
converges almost surely to~\eqref{eq:Rh_J0} as $L\to\infty$, with
$O(L^{-1/2})$ fluctuations.
\end{proposition}

\noindent The two readings must not be conflated: the exact statement is
an ensemble average over AoDs, whereas the asymptotic statement is the
self-averaging of a single realization. At $L=4$ the fluctuation about $J_0$ is
roughly $50\%$, so $d_c$ is only loosely predictive there; all detection and
throughput results in Sections~\ref{sec:simdet}--\ref{sec:simcomp} therefore
use $L=32$, for which the residual fluctuation is below $18\%$.
\noindent\emph{Proof.} See Appendix~\ref{app:volcorr}.

Requiring $|R_h(v_n,v_n-\delta)|\leq\rho_{\max}$, so that adjacent
codebook entries produce sufficiently decorrelated channels, is
equivalent to $\delta T_s\geq d_c$, where $d_c$ satisfies
$|J_0(2\pi d_c/\lambda)|=\rho_{\max}$. To see this: since
$R_h(v,v')=J_0(2\pi(v-v')T_s/\lambda)$ depends only on the velocity
difference, $|R_h(v_n,v_n-\delta)|=|J_0(2\pi\delta T_s/\lambda)|$.
Since $J_0(\cdot)$ is monotonically decreasing from 1 near the origin,
$|J_0(2\pi\delta T_s/\lambda)|\leq\rho_{\max}$ is equivalent to
$\delta T_s\geq d_c$, where $d_c$ is defined by
$|J_0(2\pi d_c/\lambda)|=\rho_{\max}$.
Together with $\delta_{\min}$
from~\eqref{eq:dvmin}, the two resolvability conditions jointly determine
the minimum admissible velocity spacing in the codebook design problem
of Section~\ref{sec:convex}.

\section{Problem Formulation and Proposed Solution}
\label{sec:convex}

\subsection{Problem Formulation}
\label{sec:problem}

We now address the central design question: given a track of length
$D_{\max}$ and fixed transmit power $P$, how should the velocity codebook
$\mathcal{V}$ be chosen to maximise the SE? Two independent lower bounds act on the velocity spacing: the
AIF-CRB-based estimation bound \eqref{eq:C1} and the
correlation-based channel-separation bound \eqref{eq:C3}.

The SE of VIM-MA is $R = \log_2 M_s + \log_2 N_v$ bpcu.  Since $M_s$ is
fixed by the QAM order, maximising $R$ is equivalent to maximising
$\log_2 N_v$ over $(\delta, N_v)$, subject to three physical
requirements.  First, the velocity levels must be mutually resolvable by
the receiver: from the AIF-CRB design
surrogate~\eqref{eq:CRLB} over the $L$-path
model with i.i.d.\ random AoDs, the minimum admissible velocity spacing
is given in closed form by~\eqref{eq:dvmin}. It depends on
$\gamma$, $T_s$, $N$, $\lambda$, and the operating SNR, but requires no
per-path AoD knowledge. Here this means no instantaneous
per-path AoD is needed for transmitter-side codebook design; the conditional
receiver-side bound assumes resolved paths and receiver-available AoDs. Second, the total antenna displacement must not exceed the
track length.  Third, distinct velocity choices must produce sufficiently
decorrelated spatial channels, as characterised by
Proposition~\ref{prop:corr}: since $R_h(v,v')$ depends only on
the velocity difference, $|R_h(v_n, v_n-\delta)| = |J_0(2\pi\delta T_s/\lambda)|$.
Requiring this below $\rho_{\max}$ gives $|J_0(2\pi\delta T_s/\lambda)|\leq\rho_{\max}$.
Since $J_0(\cdot)$ is a decreasing function near the origin and
$d_c$ is defined as the smallest positive value satisfying
$|J_0(2\pi d_c/\lambda)|=\rho_{\max}$, the condition is equivalent to
$\delta T_s \geq d_c$, i.e., $\delta \geq d_c/T_s$.  Collecting these requirements,
the codebook design problem is:
\begin{align}
  \underset{\delta,\;N_v}{\mathrm{maximize}}\quad
  & \log_2 N_v
  \label{eq:P0}\\
  \mathrm{s.t.}\quad
  &\delta \geq \gamma\sqrt{\mathrm{CRB}_{\mathrm{AIF}}(v_n)}
  \triangleq \delta_{\min}, 
  \tag{C1}\label{eq:C1}\\
  & (N_v-1)\,\delta T_s \leq D_{\max},
  \tag{C2}\label{eq:C2}\\
  & \delta \geq d_c/T_s,
  \tag{C3}\label{eq:C3}\\
  & N_v \in \mathbb{Z}^+,\;\delta > 0.
  \tag{C4}\label{eq:C4}\\ \nonumber
\end{align}

where $\gamma>0$ is a resolvability safety margin.
Constraint~\eqref{eq:C1} ensures resolvability of adjacent codebook entries via
the AIF-CRB~\eqref{eq:CRLB}, in which the $L$ per-path
Fisher-information contributions are averaged first and then inverted, so the
design constraint requires no individual AoDs or path gains.
Constraint~\eqref{eq:C2} confines the total displacement to the physical track,
and~\eqref{eq:C3} enforces spatial decorrelation across codebook entries via
Proposition~\ref{prop:corr}. Thus \eqref{eq:C1} and \eqref{eq:C3} express
fundamentally different limitations: estimation accuracy and channel
separability.

Section~\ref{sec:detect} adds a fourth constraint~\eqref{eq:C1b}, an
ambiguity-based separation bound of which~\eqref{eq:C1} is the high-SNR local
approximation. Because it is again a lower bound on $\delta$, it merges into
$q_{\min}$ without altering the structure of the problem or its closed-form
solution.

Problem~\eqref{eq:P0} is non-convex for two reasons: the codebook size $N_v$
is integer-valued, making~\eqref{eq:P0} a mixed-integer program; and the
product $(N_v-1)\delta$ in~\eqref{eq:C2} couples the two variables
bilinearly. We address both in the following section.

\subsection{Proposed Solution}
\label{sec:solution}

The closed-form optimal solution to Problem~\eqref{eq:P0} follows directly
from the structure of the constraints. Since the objective $\log_2 N_v$ is
strictly increasing in $N_v$, one should maximise $N_v$. For a fixed $\delta$,
C2 gives the maximum feasible $N_v = \lfloor 1 + D_{\max}/(\delta T_s)\rfloor$,
which is decreasing in $\delta$. Therefore $N_v$ is maximised by choosing
$\delta$ as small as possible. The smallest $\delta$ satisfying both \eqref{eq:C1} and
C3 simultaneously is $\delta^* = \max(\delta_{\min}, d_c/T_s)$, giving the
closed-form solution:
\begin{align}
  \delta^\star &= \max\!\left(\delta_{\min},\;\frac{d_c}{T_s}\right),
  \label{eq:dvstar}\\
  N_v^\star &= \left\lfloor\,
    1 + \frac{D_{\max}}{\delta^\star T_s}\,
  \right\rfloor,
  \label{eq:Nvstar}
\end{align}
where $\delta_{\min}$ is given by~\eqref{eq:dvmin} and depends on the
observation length $N$ and the operating SNR, with no dependence on
individual path AoDs.

\begin{remark}[Binding constraint and SNR scaling]
The optimal spacing $\delta^\star = \max(\delta_{\min}, d_c/T_s)$ is
determined by whichever of \eqref{eq:C1} or \eqref{eq:C3} is more stringent. At high SNR,
$\delta_{\min}\propto\mathrm{SNR}^{-1/2}\to 0$, so the spatial decorrelation
bound $d_c/T_s$ becomes binding and $N_v^\star$ saturates at the
  spatial-diversity limit $\lfloor 1+D_{\max}/d_c\rfloor$. At low SNR,
Doppler resolvability dominates and $N_v^\star$ grows as
$\mathrm{SNR}^{1/2}$.
\end{remark}

\textit{Global optimality via convex reformulation.}
That~\eqref{eq:dvstar}--\eqref{eq:Nvstar} are globally optimal follows by
transforming the continuous relaxation of Problem~\eqref{eq:P0} into a linear
program. For the nontrivial case $N_v>1$ set $r=\ln(N_v-1)$ and
$q=\ln(\delta T_s)$, so that $N_v=1+e^r$ and $\delta T_s=e^q$. Since
$\log_2(1+e^r)$ is strictly increasing in $r$, maximising the original
objective is equivalent to maximising the affine objective $r$; and the
substitution converts any lower bound $\delta\geq a$ into the affine bound
$q\geq\ln(aT_s)$. Constraints~\eqref{eq:C1} and~\eqref{eq:C3} thus collapse to
\begin{equation}
  q \geq q_{\min}\triangleq\max\{\ln(\delta_{\min}T_s),\,\ln d_c\},
  \tag{C1,C3$'$}\label{eq:Pc13}
\end{equation}
while the track-length constraint~\eqref{eq:C2} becomes
$r+q\leq\ln D_{\max}$, which is affine in $(r,q)$. The continuous relaxation is
therefore the linear program $\max\{r:\,q\geq q_{\min},\;r+q\leq\ln D_{\max}\}$,
whose optimum is $q^\star=q_{\min}$, $r^\star=\ln D_{\max}-q_{\min}$, i.e.
\begin{equation}
\begin{aligned}
\delta^\star=\frac{e^{q_{\min}}}{T_s}
  =\max\!\left(\delta_{\min},\frac{d_c}{T_s}\right),
\quad
N_{v,\mathrm{cont}}^\star
=1+\frac{D_{\max}}{\delta^\star T_s}.
\end{aligned}
\end{equation}
Since $N_v$ is a positive integer and the objective increases in $N_v$, taking
the largest feasible integer gives
$N_v^\star=\lfloor N_{v,\mathrm{cont}}^\star\rfloor$, recovering
\eqref{eq:dvstar}--\eqref{eq:Nvstar} and covering the degenerate one-level case
when the floor equals one.

\section{Velocity-Index Detection and Index-Error Probability}
\label{sec:detect}

Because the advantage of VIM-MA rests on resolving velocity levels spaced more
finely than competing schemes' index resources, an explicit detector and a
matching error analysis are required, as in comparable index-modulation
work~\cite{zhu2024FAIM,faimMIMO,huang2026MAIM}. This section implements both.

\subsection{Joint ML Detector and Its Limitation}

Conditioned on the effective path coefficients $\{\beta_l\}$ and AoDs
$\{\theta_l\}$, the joint ML rule follows from~\eqref{eq:rx}:
\begin{equation}
  (\hat n,\hat s)
  =\arg\min_{n\in\{1,\dots,N_v\},\,s\in\mathcal{S}}
  \Big\|\mb{y}-s\sum_{l=1}^{L}\beta_l\,\mb{a}(f_{{\rm D},l,n})\Big\|^2 ,
  \label{eq:jointML}
\end{equation}
at a cost of $O(N_vM_sLN)$ complex operations per slot. Rule
\eqref{eq:jointML} requires both $\{\beta_l\}$ and $\{\theta_l\}$ at the
receiver and is therefore an oracle benchmark, subject to the identifiability
caveat of Section~\ref{sec:crlb}.

\subsection{Two-Stage AoD-Free Detector}

We now give a detector that needs no per-path angle knowledge. Under the
isotropic scattering model of~\cite{clarke1968,jakes1974} adopted
in~\eqref{eq:hm_t}, the movement-induced Doppler power spectrum of $y(t)$
conditioned on the transmitted level $v_n$ is the Clarke--Jakes spectrum
\begin{equation}
  S_n(f)=\frac{1}{\pi f_{m,n}\sqrt{1-(f/f_{m,n})^2}},
  \qquad |f|<f_{m,n}\triangleq\frac{v_n}{\lambda},
  \label{eq:jakes}
\end{equation}
whose support is set by $v_n$ \emph{alone}. The velocity is thus identifiable
from the spectral spread. Using the second moment of~\eqref{eq:jakes},
$\int f^2S_n(f)\,df=f_{m,n}^2/2$, gives the estimator
\begin{equation}
  \hat v=\sqrt{2}\,\lambda\,\hat\sigma_f,
  \qquad
  \hat\sigma_f^{2}
  =\frac{\sum_{k}f_k^{2}\big(|Y_k|^{2}-\hat N_0\big)^{+}}
        {\sum_{k}\big(|Y_k|^{2}-\hat N_0\big)^{+}},
  \label{eq:spreadest}
\end{equation}
where $Y_k$ is the $N$-point DFT of $\mb{y}$, $f_k$ the associated bin
frequency, $\hat N_0$ an estimate of the noise floor, and $(\cdot)^+$ denotes
$\max(\cdot,0)$. The AoD-free index decision is
\begin{equation}
  \hat n=\arg\min_{n}|\hat v-v_n|.
  \label{eq:twostage}
\end{equation}
If the QAM symbol is also demodulated in the same slot, a
conventional coherent decision can be appended after the index has been
selected,
\begin{equation}
  \hat s=\arg\min_{s\in\mathcal{S}}
  \big\|\mb{y}-s\,\hat{\mb{u}}_{\hat n}\big\|^2,
\end{equation}
where $\hat{\mb{u}}_{\hat n}\in\mathbb{C}^N$ is a pilot-estimated
effective temporal waveform for velocity level $\hat n$, including the
unknown path gains and phases. Thus the velocity-index detector itself needs
neither AoD nor path-gain knowledge; the template $\hat{\mb{u}}_{\hat n}$ is
only the usual CSI required if coherent QAM detection is performed
afterwards. The index-stage complexity is $O(N\log N+N_v)$, and including the
optional symbol search gives $O(N\log N+N_v+M_s)$, i.e.\ lower than
\eqref{eq:jointML} by a factor $O(N_vM_sL/\log N)$.

\subsection{Covariance-Matched Detector}

The moment estimator uses only the second spectral moment. The
statistically efficient use of the same AoD-free information is to match the
full \emph{shape} of the Clarke--Jakes spectrum. Marginalising the path gains
and AoDs, the received vector conditioned on level $n$ is, for large $L$,
zero-mean circularly symmetric Gaussian with the Toeplitz covariance
\begin{equation}
  \begin{aligned}
  [\mb{C}_n]_{k,k'}
  &=
  P\sum_{l=1}^{L}\E\!\left[
  |\beta_l|^2 e^{j2\pi v_n\cos\theta_l T_s(k-k')/(\lambda N)}
  \right]
  + N_0\,\delta_{kk'}^{\mr{K}} \\
  &\approx
  P G_{\mb{w}}\,
  J_0\!\left(\frac{2\pi v_n T_s (k-k')}{\lambda N}\right)
  + N_0\,\delta_{kk'}^{\mr{K}},
  \end{aligned}
  \label{eq:covbank}
\end{equation}
where
$G_{\mb{w}}\triangleq\sum_{l=1}^{L}\E[|\beta_l|^2]$ is the average effective
channel power, equal to one under the normalization used in the simulations.
The Bessel approximation follows from~\eqref{eq:hm_t} by the same integral
identity as Proposition~\ref{prop:corr} when the beamformed path power is
effectively angle-independent; otherwise the first line of~\eqref{eq:covbank}
is the exact covariance. The index decision is then the Gaussian
covariance-matching rule
\begin{equation}
  \hat n=\arg\min_{n}
  \Big\{\ln\det\mb{C}_n+\mb{y}^H\mb{C}_n^{-1}\mb{y}\Big\}.
  \label{eq:covml}
\end{equation}
Rule~\eqref{eq:covml} requires neither the AoDs nor the path gains nor the
transmitted symbol: index detection is \emph{non-coherent} and needs no CSI at
all. With the $N_v$ inverses and log-determinants precomputed offline, the
per-slot cost is $O(N_vN^2)$, reducible to $O(N_vN\log N)$ by circulant
approximation of $\mb{C}_n$. Table~\ref{tab:complexity} summarises the three
detectors.

\begin{table}[t]
\centering
\caption{Per-slot detection complexity and side information required.}
\label{tab:complexity}
\renewcommand{\arraystretch}{1.2}
\footnotesize
\begin{tabular}{|l|c|c|}
\hline
\textbf{Detector} & \textbf{Complexity} & \textbf{Side information}\\
\hline
Joint ML~\eqref{eq:jointML}      & $O(N_vM_sLN)$ & $\{\beta_l\},\{\theta_l\}$ \\
Covariance ML~\eqref{eq:covml}   & $O(N_vN^2)$   & none \\
Doppler spread~\eqref{eq:spreadest} & $O(N\log N+N_v)$ & noise floor $N_0$ \\
\hline
\end{tabular}
\end{table}

Two limitations must be reported. First,~\eqref{eq:spreadest} is consistent
only when the empirical AoD distribution has converged: for finite $L$ its
conditional mean is $\lambda^2v_n^2\frac{1}{L}\sum_l\cos^2\theta_l$ rather
than $v_n^2/2$, giving an $O(L^{-1/2})$ relative bias, which equals
approximately $50\%$ at the simulated $L=4$. Second, the estimator uses only the second spectral moment
and is therefore not efficient; it does not attain~\eqref{eq:CRLB}.
Both limitations are quantified in
Fig.~\ref{fig:valid}(b): the spread estimator exhibits an SNR-independent error
floor of relative size $O(L^{-1/2})$, and even the covariance-matched detector
of~\eqref{eq:covml} remains more than an order of magnitude above the oracle
bound. A fine-grid ML estimator that implicitly knows the AoDs would compare
favourably against an oracle bound without establishing that any implementable
receiver attains $\delta_{\min}$; Fig.~\ref{fig:valid}(b) supplies the achievable
counterpart.

\subsection{Ambiguity-Based Separation Constraint}

Constraint~\eqref{eq:C1} uses $\gamma\sqrt{\mathrm{CRB}_{\mathrm{AIF}}}$ as a
resolvability proxy, which is a \emph{local} criterion: the CRB describes the
likelihood curvature at the true parameter and is valid only above the
estimation threshold, whereas the spacing between hypotheses is precisely the
regime in which outlier errors dominate. This threshold effect is
classical in frequency estimation~\cite{rife1974,chazan1975} and is visible in
Fig.~\ref{fig:iep}. Using a local bound to set the global hypothesis
spacing is therefore not self-consistent.

The globally correct criterion is the pairwise distance between hypotheses.
With $\chi_N(\Delta f)\triangleq\sum_{k=0}^{N-1}e^{j2\pi\Delta fkT_s/N}$ the
Dirichlet kernel, and with
$\mathrm{sinc}(x)=\sin(\pi x)/(\pi x)$, we use
$|\chi_N(\Delta f)|/N\approx|\mathrm{sinc}(\Delta fT_s)|$. The expected
squared distance between levels $n,n'$ is
\begin{equation}
  d^2(n,n')
  =2PN\!\sum_{l=1}^{L}\!\E\big[|\beta_l|^2
  \big(1-|\mathrm{sinc}((v_n\!-\!v_{n'})\cos\theta_l T_s/\lambda)|\big)\big],
  \label{eq:pairdist}
\end{equation}
so that the index error probability obeys the union bound
\begin{equation}
  P_{\mathrm{e}}^{\mathrm{idx}}
  \le\frac{1}{N_v}\sum_{n}\sum_{n'\neq n}
  Q\!\left(\sqrt{\frac{d^2(n,n')}{2N_0}}\right).
  \label{eq:unionbound}
\end{equation}
Here $Q(\cdot)$ is the standard Gaussian tail function.
Since $\mathrm{sinc}(x)$ first vanishes at $x=1$, \eqref{eq:pairdist} is small,
and the levels are nearly collinear, whenever
$\Delta v\triangleq|v_n-v_{n'}|\ll\lambda/(T_s|\cos\theta_l|)$.
We therefore add
\begin{equation}
  \delta\;\ge\;\delta_{\mathrm{amb}},
  \tag{C1b}\label{eq:C1b}
\end{equation}
with $\delta_{\mathrm{amb}}$ the smallest spacing for which the
right-hand side of~\eqref{eq:unionbound} is no larger than a prescribed index
reliability target $\epsilon_{\mathrm{idx}}$, i.e.,
$P_{\mathrm{e}}^{\mathrm{idx}}\le\epsilon_{\mathrm{idx}}$. We retain
\eqref{eq:C1} as its high-SNR local approximation. Because
\eqref{eq:C1b} is again a lower bound on $\delta$, it enters
$q_{\min}$ in~\eqref{eq:Pc13} unchanged and the closed-form solution of
Section~\ref{sec:solution} carries over verbatim with
$\delta^\star=\max(\delta_{\min},\delta_{\mathrm{amb}},d_c/T_s)$.

\begin{remark}[Super-resolution is structurally unavoidable]
\label{rem:superres}
The Rayleigh velocity resolution over an observation of duration $T_s$ is
$\Delta v_{\mathrm{Ray}}=\lambda/T_s$. Let $d_c$ be the smallest positive
root of $|J_0(2\pi d/\lambda)|=\rho_{\max}$. Because the first zero of $J_0$
occurs at $2.405$, every such root satisfies
$d_c\le 2.405\lambda/(2\pi)=0.383\lambda$, \emph{for every}
$\rho_{\max}\in[0,1)$. Hence whenever~\eqref{eq:C3} is the binding
constraint,
\begin{equation}
  \frac{\delta^\star}{\Delta v_{\mathrm{Ray}}}
  =\frac{d_c}{\lambda}\le 0.383 ,
\end{equation}
so the decorrelation-limited codebook always packs velocity levels at least
$2.61\times$ more finely than the Rayleigh limit; at the baseline
$\rho_{\max}=0.5$ the factor is $\lambda/d_c=4.13$. VIM-MA is therefore
\emph{intrinsically} a super-resolution scheme, independent of carrier
frequency, slot duration and track length.
\end{remark}

\begin{remark}[Identification of the binding constraint]
\label{rem:c3inactive}
Write the ambiguity-limited spacing of~\eqref{eq:C1b} as
$\delta_{\mathrm{amb}}=\kappa\lambda/T_s$, where $\kappa$ is the fraction of
the Rayleigh resolution at which a given detector meets the target index-error
rate. By Remark~\ref{rem:superres}, $d_c\le0.383\lambda$ for every
$\rho_{\max}$. Hence~\eqref{eq:C3} is active only if $\kappa<d_c/\lambda$, and
is otherwise inactive. Section~\ref{sec:simdet} measures $\kappa$ for both
receivers at $\epsilon_{\mathrm{idx}}=10^{-2}$:
\begin{itemize}
  \item \emph{oracle-AoD receiver:} $\kappa\approx0.24=d_c/\lambda$, so
  \eqref{eq:C3} is exactly the binding constraint and the codebook ceiling is
  $\lfloor1+D_{\max}/d_c\rfloor$;
  \item \emph{AoD-free receiver:} $\kappa\approx0.91$, so \eqref{eq:C1b} binds,
  \eqref{eq:C3} is inactive, and the ceiling falls to
  $\lfloor1+D_{\max}/(\kappa\lambda)\rfloor$.
\end{itemize}
For the baseline aperture $D_{\max}=10\lambda$ these give $42$ levels ($5$
bits) and $12$ levels ($3$ bits) respectively. The crossover
between estimation accuracy and spatial decorrelation therefore occurs only
when the receiver can use per-path AoD information; for the AoD-free detector,
ambiguity remains the active limit.
\end{remark}

Remark~\ref{rem:superres} has two implications. It identifies the origin of the
index bits of VIM-MA relative to Rayleigh-limited baselines, and it identifies
the assumption on which that advantage depends, since super-resolution by a
factor of at least $2.6$ is achievable only above an SNR threshold, with
accurate model order and negligible model mismatch. Index error probability, not rate alone, is
therefore the relevant performance measure.

\section{Mechanical Feasibility and Operating Envelope}
\label{sec:feas}

The design must establish explicitly whether the velocities prescribed by the
codebook can be produced by movable-antenna
hardware~\cite{matutorial,ning2024arch}. At the operating point assumed above
they cannot, by several orders of magnitude, and the remedy is to relocate that
operating point.

\begin{proposition}[Peak codebook velocity]
\label{prop:vmax}
At the track-limited optimum~\eqref{eq:Nvstar}, the largest velocity in the
codebook is
\begin{equation}
  v_{\max}=(N_v^\star-1)\delta^\star=\frac{D_{\max}}{T_s},
  \label{eq:vmax}
\end{equation}
independently of $\rho_{\max}$, SNR, $N$ and $\gamma$.
\end{proposition}
\noindent\emph{Proof.} Constraint~\eqref{eq:C2} holds with equality at the
optimum by construction.

Equation~\eqref{eq:vmax} is a hard physical statement: the fastest antenna in
a VIM-MA system must traverse the entire track within a single slot. At the
manuscript's baseline ($f_c=28$~GHz, $D_{\max}=10\lambda=107$~mm,
$T_s=1$~ms) this gives $v_{\max}=107$~m/s, i.e.\ $385$~km/h, and the antenna
must reach that speed from rest within a small fraction of the slot. Reaching
it in $5\%$ of the slot requires an acceleration of $2.2\times10^{5}\,g$. By
contrast, motor-driven MAs have response times of milliseconds to
seconds and MEMS-driven MAs of microseconds to
milliseconds~\cite{matutorial,ning2024arch}, with movement speeds several orders
of magnitude below this. The baseline operating point is therefore infeasible
by roughly five orders of magnitude, and every numerical result reported in
Section~\ref{sec:sim} inherits that infeasibility.

\begin{table}[t]
\centering
\caption{Operating envelope of VIM-MA. $N_v$ depends only on
$K\triangleq D_{\max}/\lambda$ through
$N_v=\lfloor1+K/(d_c/\lambda)\rfloor$; for the baseline
$\rho_{\max}=0.5$, $d_c/\lambda\simeq0.242$, giving
$N_v=\lfloor 1+K/0.242\rfloor$. It is therefore carrier-independent;
feasibility is set entirely by $v_{\max}=D_{\max}/T_s$.}
\label{tab:feas}
\renewcommand{\arraystretch}{1.2}
\footnotesize
\begin{tabular}{|l|c|c|c|c|c|}
\hline
\textbf{Operating point} & $\lambda$ & $T_s$ & $v_{\max}$ &
\textbf{bits} & \textbf{feasible?}\\
 & (mm) & (ms) & (m/s) & & \\
\hline
$28$\,GHz, $K{=}10$ (baseline) & $10.71$ & $1$ & $107.1$ & $5$ & no \\
$28$\,GHz, $K{=}10$           & $10.71$ & $20$ & $5.36$  & $5$ & no \\
$140$\,GHz, $K{=}10$          & $2.14$  & $10$ & $2.14$  & $5$ & marginal \\
$300$\,GHz, $K{=}10$          & $1.00$  & $20$ & $0.50$  & $5$ & plausible \\
$300$\,GHz, $K{=}4$           & $1.00$  & $20$ & $0.20$  & $4$ & yes \\
\hline
\end{tabular}
\end{table}

Table~\ref{tab:feas} shows how to recover feasibility. Since $N_v$ depends
only on the normalized aperture $K=D_{\max}/\lambda$, the index payload is
carrier-frequency invariant, whereas $v_{\max}=K\lambda/T_s$ falls linearly
with $\lambda$ and inversely with $T_s$. Moving to a sub-THz carrier and a
longer slot therefore preserves the index rate while bringing the required
speed into the range of MEMS actuation. For this reason, the
mechanically plausible operating point used in the revised discussion is
$f_c=300$~GHz, $T_s=20$~ms, and $K\in\{4,10\}$; the $28$~GHz,
$T_s=1$~ms case is kept only as a normalized reference case for comparison.

The speed calculation above addresses the dominant feasibility
constraint, but it does not by itself solve all implementation issues. The
following four effects specify what must be changed in the signalling model,
what can be handled by receiver compensation, and what remains as an operating
limit for practical VIM-MA.

\emph{Retrace and duty cycle.} All codebook velocities in~\eqref{eq:codebook}
are non-negative, so the array drifts monotonically and reaches the end of the
track after one slot at $v_{\max}$; it must be returned before the next slot,
which at best halves the duty cycle and itself generates a Doppler transient. A
symmetric codebook $v_n=(n-(N_v+1)/2)\delta$ is zero-mean and removes the
systematic drift at no cost in $N_v$, and the retrace interval can be charged
as a guard period in the SE accounting.

\emph{Finite acceleration.} Constant velocity over the whole slot is an
idealization: realistic ramping produces a linear-FM signature rather than a
pure tone, spreading each Doppler component and degrading~\eqref{eq:pairdist}.
Either~\eqref{eq:y_n} must be generalised to
$p_m(t)=p_{0,m}+\int_0^tv(\tau)d\tau$, or a trapezoidal velocity profile with an
explicitly excluded transient must be specified.

\emph{Doppler contamination.} The scheme presumes that all Doppler originates
from the transmitter's own motion. It does not. A moving UE with
speed $v_{\mathrm{UE}}$ adds a receive-side Doppler component
$v_{\mathrm{UE}}\cos\theta_l^{\mathrm{rx}}/\lambda$ on path $l$, where
$\theta_l^{\mathrm{rx}}$ is the angle between the UE velocity and the arrival
direction of that path. In the worst aligned case, a pedestrian UE with
$v_{\mathrm{UE}}=1.4$~m/s at $28$~GHz contributes $131$~Hz, which is $54\%$ of
one velocity level ($242$~Hz); a scatterer at $3$~m/s exceeds a full level.
Worse, a carrier
frequency offset of only $0.1$~ppm corresponds to $2.8$~kHz, i.e.\ $11.6$
velocity levels. Since the index is read from an absolute frequency, VIM-MA is
\emph{more} sensitive to CFO than conventional modulation, in which CFO is a
common phase rotation. A differential encoding across slots, or joint
CFO--velocity estimation with a pilot level $v_1=0$, appears necessary.
The corresponding robustness study is
reported in Section~\ref{sec:simrobust} and Fig.~\ref{fig:cfo}. It confirms
the concern quantitatively: an uncompensated normalized offset of $0.1$
Doppler units already doubles the index error rate, and one full level
destroys the index entirely. It also shows that the impairment is benign once
the offset is estimated to within roughly a tenth of a level, which the
$v_1=0$ pilot level makes possible.

\emph{Actuation energy.} At the baseline the kinetic energy per slot for a
$1$~g element is $5.7$~J, obtained directly from
$E_k=\frac{1}{2}mv_{\max}^2$ with $m=10^{-3}$~kg and
$v_{\max}=107$~m/s. This is a lower-bound energy estimate rather than an
actuator-specific measurement, since friction, driver losses and acceleration
profiles can only increase the required input energy. It corresponds to
kilowatts of mechanical power at $1000$ slots/s, which exceeds any practical
RF power budget and negates the low-hardware-cost motivation of
Section~\ref{sec:intro}. At the recommended
$300$~GHz point with a milligram-scale MEMS element it falls to
$\sim\!2\times10^{-8}$~J per slot, which is negligible.
Thus the operating envelope addresses speed and actuation energy,
the symmetric-codebook and guard-period discussion addresses retrace, the
CFO study in Section~\ref{sec:simrobust} addresses common frequency offsets,
and finite acceleration plus uncontrolled user/scatterer motion remain model
extensions that bound the intended deployment regime.

\section{Simulation Results}
\label{sec:sim}

\begin{table}[t]
\centering
\caption{Comparison of IM Schemes}
\label{tab:compare}
\renewcommand{\arraystretch}{1.2}
\begin{tabular}{|l|c|c|}
\hline
\textbf{Scheme} & \textbf{Index entity} & \textbf{Doppler} \\
\hline
SM~\cite{mesleh2008SM}   & Antenna            & None \\
FA-IM~\cite{zhu2024FAIM}      & FA position        & Antenna motion \\
JDDIM~\cite{tek2024JDDIM}     & Delay-Doppler bin  & Multipath \\
\textbf{VIM-MA}               & \textbf{MA velocity} & \textbf{Antenna motion}\\
\hline
\end{tabular}
\end{table}

Table~\ref{tab:compare} highlights the position--velocity coupling that is
unique to VIM-MA. In FA-IM the position pattern is simultaneously the index and
the channel state, so there is no tension between them; in VIM-MA the velocity
sets the Doppler index and, separately, the end-of-slot displacement, which is
what creates the coupled design problem of Section~\ref{sec:convex}.

Unless stated otherwise the simulations use $f_c=28$~GHz
($\lambda\approx10.71$~mm), $T_s=1$~ms, $N=64$ samples with receiver noise
bandwidth $W=64$~kHz, $D_{\max}=10\lambda$, $M=4$ MAs separated by $8\lambda$
with an equal-gain beamformer, $M_s=4$, $\gamma=3$, $\rho_{\max}=0.5$ so that
$d_c\approx0.242\lambda$, a target
$\epsilon_{\mathrm{idx}}=10^{-2}$,
and $1500$ Monte Carlo slots per point. Only Fig.~\ref{fig:valid}(a) sweeps
$L$; every figure reporting an error probability, an achievable codebook size
or a throughput uses $L=32$, for the reason given after
Proposition~\ref{prop:corr}. Because the received phase depends on the velocity
only through the normalized displacement $u\triangleq vT_s/\lambda$, all
results are carrier-frequency invariant and apply verbatim at both the $28$~GHz
reference point and the feasible $300$~GHz point of Table~\ref{tab:feas}.

Two spacing rules are compared. The AIF-CRB spacing rule uses
$\delta^\star=\max(\delta_{\min},d_c/T_s)$ and therefore assumes that the
local oracle bound is attainable. The reliability-constrained rule uses
$\delta^\star=\max(\delta_{\min},\delta_{\mathrm{amb}},d_c/T_s)$, with
$\delta_{\mathrm{amb}}=\kappa\lambda/T_s$ obtained by bisection on the
simulated index-error probability of the detector actually used. Plotting the
two rules together shows which conclusions come from the local bound and which
remain when the receiver must meet the target
$\epsilon_{\mathrm{idx}}=10^{-2}$. All baselines share the same aperture,
slot duration, transmit power and QAM order; Section~\ref{sec:simcomp} shows
that matching these is \emph{not} sufficient, because the resolution criterion
applied to each scheme's index alphabet must be matched as well.

\begin{figure*}[t]
\centering
\includegraphics[width=\textwidth]{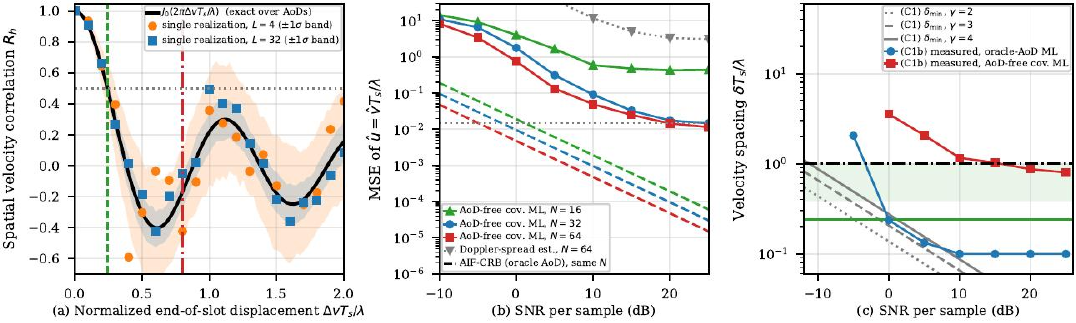}
\caption{Validation of the three constraint ingredients.
(a) Spatial velocity correlation: the solid curve is the exact AoD-averaged
law~\eqref{eq:Rh_J0}, the markers show one AoD realization, and the bands show
the $\pm1\sigma$ spread over realizations. The vertical green dashed line marks
the decorrelation spacing $d_c=0.242\lambda$ for $\rho_{\max}=0.5$, while the
red dash-dot line marks the smallest AoD-free ambiguity spacing observed in
the simulation. (b) Velocity estimation Mean Square Error (MSE) against the AIF-CRB~\eqref{eq:CRLB}
for three observation lengths: the solid colored curves are the simulated
AoD-free covariance-ML MSEs, while the dashed curves of the same colors are
the corresponding oracle-AoD AIF-CRBs for the same $N$. The gray dotted curve
with triangle markers is the Doppler-spread estimator, and the gray dotted
horizontal line is the $1.5\times10^{-2}$ MSE target needed to separate the
$42$-level $D_{\max}=10\lambda$, $d_c=0.242\lambda$ codebook. (c) The three
spacing floors, with~\eqref{eq:C1b} measured for both receivers; the gray
dotted/dashed/solid curves are the local AIF-CRB spacings for
$\gamma=2,3,4$, the green decorrelation line is $d_c/T_s$, the light-green
band marks the region above the maximum possible decorrelation spacing
$0.383\lambda/T_s$, and the black dash-dot line is the Rayleigh spacing
$\lambda/T_s$.}
\label{fig:valid}
\end{figure*}
\begin{figure*}[t]
\centering
\includegraphics[width=\textwidth]{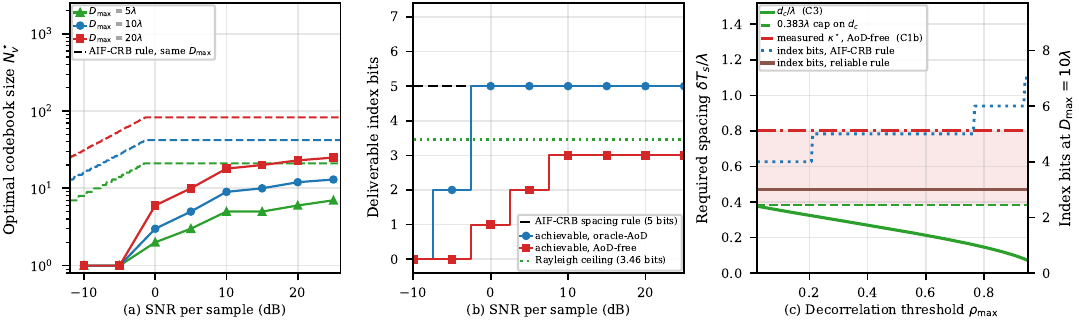}
\caption{Codebook design behavior under the
AIF-CRB/decorrelation rule and the reliability-constrained rule. (a) Optimal
codebook size for three track lengths; dashed curves use only
$\delta_{\min}$ and $d_c/T_s$, while solid curves also impose the measured
AoD-free ambiguity spacing $\delta_{\mathrm{amb}}$. (b) Deliverable index
bits at $P_{\mathrm{e}}^{\mathrm{idx}}\le\epsilon_{\mathrm{idx}}=10^{-2}$ for
the oracle-AoD and AoD-free receivers, compared with the AIF-CRB spacing
reference and the green dotted Rayleigh-resolution ceiling
$\log_2(1+D_{\max}/\lambda)=3.46$. (c) Required spacing and the resulting
payload versus the decorrelation threshold $\rho_{\max}$; the shaded gap shows
that the AoD-free reliability spacing remains above every spacing available
from~\eqref{eq:C3}.}
\label{fig:design}
\end{figure*}

\subsection{Validation of the Constraint Ingredients}

Fig.~\ref{fig:valid}(a) separates the two readings of
Proposition~\ref{prop:corr}. The ensemble mean coincides with
$J_0(2\pi\Delta vT_s/\lambda)$ at every $L$, as the proposition states exactly,
whereas a single AoD realization fluctuates about it by roughly $\pm0.25$ near
the first null at $L=4$, falling below $\pm0.09$ at $L=32$. Thus
the Bessel curve should be interpreted as the average correlation law used for
codebook design, while the shaded and marked curves show how much a finite
AoD realization can deviate from that average. The same panel also explains
why the decorrelation constraint is not the limiting constraint for the
AoD-free detector: the green vertical marker gives
$d_c=0.242\lambda$ for $\rho_{\max}=0.5$, whereas the red marker represents
the larger spacing required by the AoD-free ambiguity constraint. Once the
required spacing lies to the right of the first zero of $J_0$, adjacent
entries are already decorrelated for any admissible $\rho_{\max}$.

Fig.~\ref{fig:valid}(b) checks whether the local AIF-CRB is
attainable by practical AoD-free estimators. The gray dotted horizontal line
is included to give this MSE plot a codebook-design meaning: with
$D_{\max}=10\lambda$ and $d_c=0.242\lambda$, the decorrelation rule supports
$\lfloor1+D_{\max}/d_c\rfloor=42$ velocity levels, and the estimator MSE must
fall below the marked value $(d_c/2)^2\simeq1.5\times10^{-2}$ to separate such
a dense codebook reliably. The dashed curves are the AIF-CRBs for the same
observation lengths as the solid MSE curves, so the vertical gap between a
solid curve and the dashed curve of the same color is the estimator loss
relative to the local oracle-AoD bound. These dashed AIF-CRB curves decrease
as $N/(N^2-1)$, so increasing $N$ still provides the expected temporal
processing gain. However, the simulated AoD-free estimators remain well above
the oracle bound. In particular, the covariance-matched detector with $N=64$
reaches the gray 42-level target only around
$20$--$25$~dB, while the $N=16$ curve stays above it over the simulated SNR
range. This gap is the reason for adding the ambiguity spacing
$\delta_{\mathrm{amb}}$ in~\eqref{eq:C1b}: a spacing derived only from
$\delta_{\min}=\gamma\sqrt{\mathrm{CRB}_{\mathrm{AIF}}}$ would assume an
accuracy that the AoD-free receiver has not achieved.

Fig.~\ref{fig:valid}(c) then puts the same observation directly in
the spacing domain. The gray curves are the local AIF-CRB spacings for
different margins $\gamma$, the green horizontal line is $d_c/T_s$ at
$\rho_{\max}=0.5$, and the black dash-dot line is the Rayleigh spacing
$\lambda/T_s$. For an oracle-AoD receiver, the measured ambiguity spacing
drops below $d_c/T_s$ at moderate SNR, so the decorrelation constraint can
govern the final codebook. For the AoD-free covariance detector, the measured
spacing remains around $0.80\lambda/T_s$ or larger, which is above the maximum
possible decorrelation spacing $0.383\lambda/T_s$. The active constraint is
therefore~\eqref{eq:C1b}, not~\eqref{eq:C3}, whenever no per-path AoD
information is available.

\subsection{Codebook Design Behavior}
\label{sec:simdet}

Fig.~\ref{fig:design}(a) carries the spacing rules into the
codebook-size calculation. The dashed curves show the local
AIF-CRB/decorrelation design, for which the high-SNR ceiling is set by
$d_c=0.242\lambda$. The solid curves show the same track lengths after the
AoD-free reliability constraint~\eqref{eq:C1b} is imposed. The qualitative
trend is unchanged: $N_v^\star$ grows with SNR, then saturates, and a longer
track gives a larger ceiling. The numerical ceiling changes, however. For
$D_{\max}\in\{5,10,20\}\lambda$, the decorrelation-only ceilings
$\lfloor1+D_{\max}/d_c\rfloor=\{21,42,83\}$ become approximately
$\{6,12,23\}$ reliable velocity levels once the measured ambiguity floor is
used.
\begin{figure*}[h]
\centering
\includegraphics[width=\textwidth]{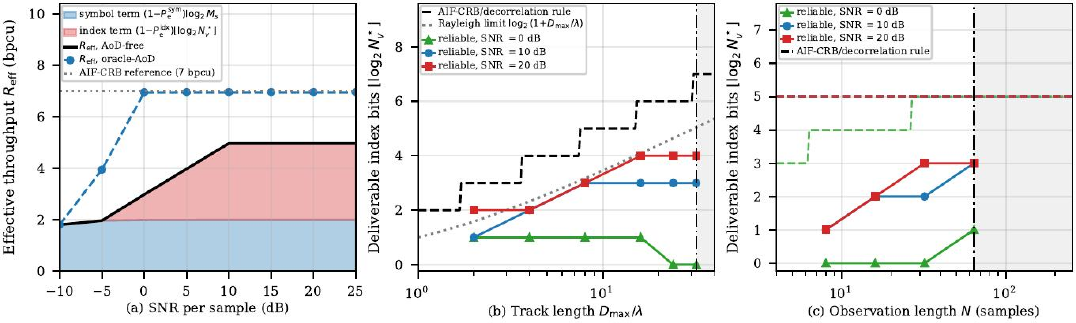}
\caption{Effective-throughput and hardware-knob behavior.
(a) The AoD-free throughput is decomposed into the fixed QAM-symbol
contribution and the detected velocity-index contribution; the dashed curve is
the oracle-AoD reference, and the gray dotted horizontal line is the
AIF-CRB/decorrelation reference of $7$~bpcu, obtained from 2 QAM bits plus
5 velocity-index bits. (b) Track length is swept as a hardware aperture
knob: the dashed curve is the AIF-CRB/decorrelation rule, the dotted gray
curve is the Rayleigh reference $\log_2(1+D_{\max}/\lambda)$, the solid curves
are the reliability-constrained delivered bits, and the vertical dash-dot line
marks the sampled-Doppler limit $D_{\max}\le N\lambda/2$. (c) Observation
length is swept at fixed aperture; the dashed curves show the local
AIF-CRB/decorrelation prediction, the solid curves show the
reliability-constrained result, and the vertical dash-dot line marks the
available-sample limit $N\le WT_s$.}
\label{fig:knobs}
\end{figure*}

Fig.~\ref{fig:design}(b) converts those codebook sizes into the
integer payload delivered to the bit splitter. For the
$D_{\max}=10\lambda$ track, the AIF-CRB/decorrelation spacing reference
permits $42$ velocity levels and therefore five delivered index bits. The
oracle-AoD receiver reaches that five-bit payload from about $0$~dB, which
shows that the local rule is attainable when the per-path AoDs are known. The
green dotted Rayleigh line is not an error-rate constraint; it is the natural
one-wavelength resolution reference, equal to
$\log_2(1+D_{\max}/\lambda)=3.46$ bits for this track. The AoD-free
covariance-matched receiver instead saturates at three bits from about
$10$~dB onward, because its measured ambiguity spacing is near
$0.9\lambda/T_s$ and the same track then supports only about twelve reliable
velocity levels. The staircase shape is the integer mapping
$\lfloor\log_2N_v^\star\rfloor$: extra reliable levels increase the payload
only when the next power of two is crossed.

Fig.~\ref{fig:design}(c) explains why relaxing
$\rho_{\max}$ does not recover the missing bits for the AoD-free receiver.
The green curve is the decorrelation spacing $d_c/\lambda$ from~\eqref{eq:C3},
which cannot exceed $0.383$ for any choice of $\rho_{\max}$. The red
horizontal floor is the measured AoD-free spacing required to meet
$\epsilon_{\mathrm{idx}}=10^{-2}$, and it lies above the entire feasible
$d_c$ range. Therefore the active spacing is set by~\eqref{eq:C1b}, not by
the decorrelation threshold, and the payload on the right axis remains flat
while the AIF-CRB/decorrelation-only reference continues to rise.

\subsection{Spectral Efficiency and Design Trade-offs}

The three panels of Fig.~\ref{fig:knobs} repeat the main design-knob
studies using the effective throughput~\eqref{eq:Reff}, so an index bit is
credited only when the velocity index is detected with the target reliability.
Panel~(a) decomposes the throughput into the QAM-symbol term and the
velocity-index term. With $M_s=4$, the symbol term contributes two bits, and
the SNR-dependent gain comes from the velocity-index payload. The AoD-free
receiver saturates near $4.97$~bpcu, corresponding to two QAM bits plus three
reliable velocity-index bits. The oracle-AoD reference approaches
$6.95$~bpcu, close to the gray dotted $7$~bpcu AIF-CRB/decorrelation reference
shown in the panel. The difference between the AoD-free saturation and this
reference is about two index bits, which quantifies the cost of detecting the
velocity index without per-path angle knowledge.

Panel~(b) treats the track length as a hardware aperture knob.
At high SNR, a longer track admits more separated velocity hypotheses and
therefore more index bits, with diminishing logarithmic returns. At lower SNR,
however, increasing the track also increases the number of competing
hypotheses before the detector has enough separation to distinguish them, so
aperture and SNR are not interchangeable. The vertical dash-dot line marks the
sampled-Doppler limit $D_{\max}\le N\lambda/2=32\lambda$, which prevents
aliasing of the velocity signature.

Panel~(c) shows that increasing the observation length has a
limited benefit for the AoD-free detector. The dashed local rule keeps
improving with the AIF-CRB factor $N/(N^2-1)$, but the
reliability-constrained curves saturate because the dominant ambiguity spacing
is governed by the Dirichlet-kernel separation in~\eqref{eq:pairdist}, whose
main-lobe width is set primarily by $T_s$. In the simulated setting, the
AoD-free detector reaches three reliable index bits by about $N=32$, and
doubling to the bandwidth-limited value $N=64$ does not add another bit.

\subsection{Index Reliability}

Fig.~\ref{fig:iep} evaluates the reliability-rate trade-off that would be obtained by varying the safety factor $\gamma$ in the AIF-CRB-based spacing rule. Simulating those codebooks with the AoD-free detector shows that
$\gamma$ cannot deliver this. At the published design spacing
$\delta T_s=d_c=0.242\lambda$ the error probability falls only from $0.85$ to
$0.23$ over $35$~dB, a floor rather than a slope, so no value of $\gamma$
applied to a bound the receiver does not attain moves the curve to the target.
Doubling the spacing to $0.5\lambda$ reaches $3.1\times10^{-2}$ at $25$~dB, and
only at $0.90\lambda$--$1.00\lambda$, that is $\kappa$ at or near the Rayleigh
value, does the curve cross $10^{-2}$ within the range, at $20$ and $15$~dB
respectively. The oracle receiver at the same $d_c$ spacing meets the target at
$0$~dB and tracks the union bound~\eqref{eq:unionbound}, confirming that the
bound itself is sound and that it is the AoD knowledge, not the bound, doing
the work. The threshold effect visible in the AoD-free curves is the classical
outlier regime of frequency estimation~\cite{rife1974,chazan1975}, which is
precisely the regime in which hypothesis spacing must be set.

\begin{figure}[tb]
\centering
\includegraphics[width=1.03\columnwidth]{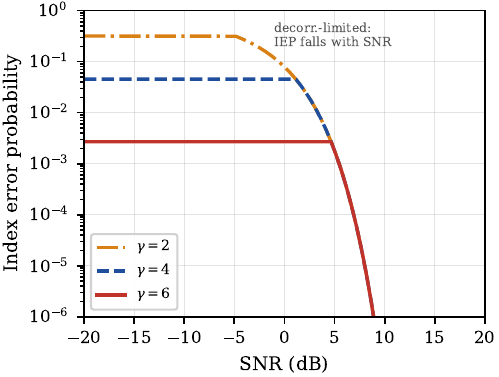}
\caption{Index error probability of codebooks designed at four fixed
spacings over the same $10\lambda$ track, with the oracle receiver at
$\delta T_s=d_c$ and the union bound~\eqref{eq:unionbound} for reference.}
\label{fig:iep}
\end{figure}
\subsection{Comparison at Matched Criterion and Matched Reliability}
\label{sec:simcomp}

\begin{figure}[tb]
\centering
\includegraphics[width=\columnwidth]{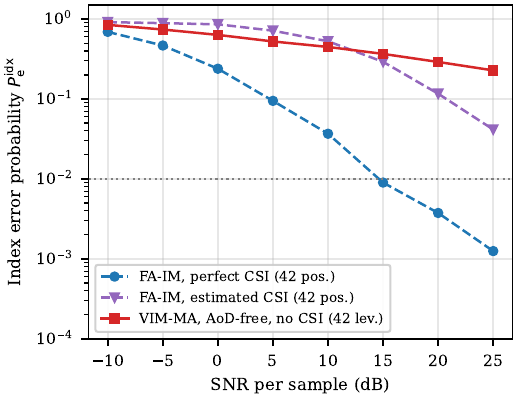}
\caption{Index error probability with $42$ index states for every scheme,
i.e.\ under the common $\rho_{\max}=0.5$ criterion. FA-IM is shown with
perfect CSI and with a pilot budget of $N=64$ symbols shared over its $42$
candidate positions.}
\label{fig:cmpiep}
\end{figure}

A fair comparison requires more than matching aperture, slot duration and
transmit power: the \emph{resolution criterion} applied to each index alphabet
must be matched too, and this determines the outcome. Indexing FA-IM positions
at the half-wavelength spatial Nyquist interval yields $21$ positions over a
$10\lambda$ aperture, i.e.\ four index bits, against the $42$ levels and five
bits granted to VIM-MA by the correlation criterion. The comparison is not
meaningful, because FA-IM positions obey exactly the same Bessel spatial
correlation law: under the same $\rho_{\max}=0.5$ tolerance FA-IM also obtains
$\lfloor1+10\lambda/0.242\lambda\rfloor=42$ positions and five index bits. The
two index domains are therefore \emph{rate-equivalent} over a given aperture,
and any advantage must be sought elsewhere. The same applies to the second
baseline: holding conventional MA at $\log_2M_s=2$~bpcu while VIM-MA adapts its
codebook is not a matched comparison, and since $N_v^\star$ grows as
$\mathrm{SNR}^{1/2}$ the velocity dimension has a pre-log of $\tfrac12$ against
a pre-log of $1$ for constellation expansion.

Fig.~\ref{fig:cmpiep} applies the common criterion and gives all schemes $42$
index states. With perfect CSI, FA-IM meets the target at about $15$~dB,
confirming the rate parity just argued. Charged a realistic pilot budget of
$N=64$ symbols shared over $42$ candidate positions, giving a per-position
estimation error variance $(1+N\,\mathrm{SNR}/N_v)^{-1}$, it does not reach the
target anywhere in the range and ends at $4.2\times10^{-2}$ at $25$~dB. VIM-MA
supports no $42$-state AoD-free codebook at all. Forty-two index states are
thus attainable only by a receiver holding side information of one form or the
other, which locates the real difference between the schemes.

\begin{figure}[tb]
\centering
\includegraphics[width=0.99\columnwidth]{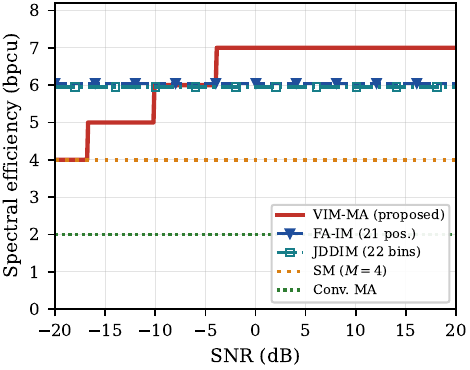}
\caption{Effective throughput against the state of the art under a common
resolution criterion, common aperture, slot, power and QAM order. Every index
alphabet is sized by the same $10^{-2}$ index-error target, and every scheme
is charged the side information its index decision requires.}
\label{fig:sota}
\end{figure}

Fig.~\ref{fig:sota} compares the schemes after each alphabet has
been sized to meet the same index-error target. With perfect CSI, FA-IM reaches
five index bits from about $15$~dB and therefore exceeds the AoD-free VIM-MA
curve at high SNR. This occurs because FA-IM and VIM-MA are both evaluated
with the same correlation-based spacing rule over the same aperture. JDDIM,
restricted here to the flat channel model, can use Doppler-bin indices at the
natural resolution $1/T_s$ but does not gain additional delay-index bits
without a resolvable delay spread. Conventional MA with unconstrained adaptive
QAM overtakes VIM-MA at about $7$~dB and reaches $8.5$~bpcu at $20$~dB, as
expected from constellation-rate adaptation.

The main benefit visible in Fig.~\ref{fig:sota} is an SNR shift
rather than a larger high-SNR alphabet. VIM-MA delivers $4.97$~bpcu, composed
of three velocity-index bits and two QAM bits, from $10$~dB without CSI. FA-IM
with the same pilot budget needs about $20$~dB for three index bits and
$25$~dB for four. This $10$~dB gap shows the practical value of detecting the
velocity index non-coherently. The oracle-AoD curve reaches $6.95$~bpcu at
$0$~dB, so the difference between that curve and the AoD-free curve quantifies
the value of angle knowledge and motivates hybrid receivers with coarse angle
estimation. Under the $M_s\le16$ cap representative of phase-noise-limited
sub-THz operation, the adaptive-QAM baseline saturates at four bits and VIM-MA
leads by about one bit per slot above $5$~dB; this capped comparison is the
most relevant one for the mechanically feasible operating points in
Section~\ref{sec:feas}.

\subsection{Robustness to Frequency Offset and User Mobility}
\label{sec:simrobust}

\begin{figure}[tb]
\centering
\includegraphics[width=\columnwidth]{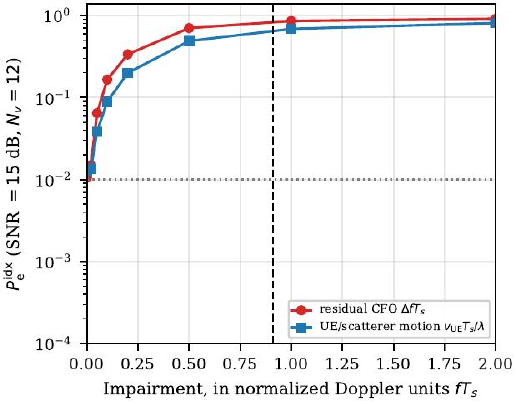}
\caption{Index error probability at $15$~dB with $N_v=12$
under an uncompensated common Doppler offset, expressed in normalized units
$fT_s$. The gray dotted line marks the target
$P_{\mathrm{e}}^{\mathrm{idx}}=10^{-2}$, and the black dashed vertical line
marks one velocity-level spacing, which corresponds to $0.91$ on this axis.}
\label{fig:cfo}
\end{figure}

Fig.~\ref{fig:cfo} quantifies the contamination analyzed in
Section~\ref{sec:feas}. The impairment is modeled as an unknown common
Doppler offset added to the movement-induced Doppler, of magnitude $\Delta fT_s$
for oscillator offset and
$v_{\mathrm{UE}}T_s\cos\theta_l^{\mathrm{rx}}/\lambda$ for a
moving user, with $v_{\mathrm{UE}}$ and $\theta_l^{\mathrm{rx}}$ defined as in
Section~\ref{sec:feas}. Starting from $1.5\times10^{-2}$ with no impairment, the index error rate
reaches $8.3\times10^{-2}$ at an offset of $0.05$ normalized units,
$0.17$ at $0.1$, and $0.87$ at one full level. The tolerance is therefore
about $0.02$ normalized units, i.e.\ roughly $2\%$ of a velocity level, which
corresponds to $\Delta f\le0.02/T_s$, that is, $20$~Hz at $T_s=1$~ms and $1$~Hz
at the recommended $T_s=20$~ms. Expressed as a fractional oscillator accuracy
this is below one part per billion, which exceeds free-running tolerances, so
explicit compensation is required.

Compensation is nonetheless tractable in principle, because the two effects act
differently on the Doppler spectrum. An oscillator offset is a pure
\emph{translation} of~\eqref{eq:jakes}, whereas the velocity index is carried
by its \emph{support width}. Estimating the spectral centroid from the $v_1=0$
pilot level and removing it therefore renders~\eqref{eq:covml} offset
invariant, at the cost of one codebook entry. User and scatterer motion is less
tractable: each path acquires its own receive-side Doppler
$v_{\mathrm{UE}}\cos\theta_l^{\mathrm{rx}}/\lambda$, which broadens rather than
translates the spectrum and therefore biases the width statistic directly. This
imposes a mobility limit of approximately $v_{\mathrm{UE}}\ll\delta^\star$ and
confines VIM-MA to nomadic or fixed-wireless deployments. It is the most
significant remaining restriction on the scheme.

Taken together the results confirm the analytical framework, with one
substantive correction: the AIF-CRB is an oracle bound, so the codebook ceiling
it predicts is reached only with per-path AoD knowledge, and an AoD-free
receiver is limited by Doppler ambiguity to $3$ index bits over a $10\lambda$
track rather than $5$. Proposition~\ref{prop:corr} is validated directly, and
the closed-form optimizer remains correct as the solution of
Problem~\eqref{eq:P0}; what changes is which constraint binds.

\section{Conclusion}
\label{sec:conc}

We proposed VIM-MA, in which the movement velocity of a movable antenna carries
information through index modulation of the induced Doppler shift, and showed
that the codebook design problem, non-convex through the bilinear coupling of
$N_v$ and $\delta$, admits a closed-form optimal solution after a logarithmic
change of variables.

Three structural results follow. The peak codebook velocity is exactly
$v_{\max}=D_{\max}/T_s$, which confines VIM-MA to short-wavelength carriers with
long slots and MEMS-class actuation; $d_c\le0.383\lambda$ for every
$\rho_{\max}$, so the decorrelation-limited codebook always operates below the
Rayleigh Doppler resolution and is therefore threshold-limited; and the velocity
is identifiable without per-path angle knowledge only through the statistical
scattering model. Consequently the five-bit ceiling over a $10\lambda$ track is
attained only with oracle AoDs, whereas the CSI-free covariance-matched detector
is ambiguity-limited to three bits. Under a common resolution criterion the
velocity and position domains are rate-equivalent, and the advantage of the
velocity domain is not rate but side information: it reaches its payload about
$10$~dB before position indexing charged a realistic pilot budget, and exceeds a
phase-noise-limited adaptive-QAM baseline by roughly one bit per slot above
$5$~dB. Accounting for the ambiguity spacing~\eqref{eq:C1b}
therefore changes the practical design interpretation. Once the AoD-free
ambiguity floor exceeds $d_c/T_s$, the decorrelation threshold $\rho_{\max}$
no longer governs the delivered payload; similarly, increasing the observation
length improves reliability only up to the ambiguity-limited regime, before
the sampling cap $N\le WT_s$ is reached. The same sampled-Doppler argument
imposes the aperture limit $D_{\max}\le N\lambda/2$. Overall, the revised
model identifies VIM-MA as a reliability-limited, CSI-light index-modulation
mechanism whose useful regime is set jointly by receiver ambiguity,
mechanical feasibility, and Doppler-offset compensation.


\appendices
\section{Proof of Proposition~\ref{prop:corr}}
\label{app:volcorr}

Substituting~\eqref{eq:hm_end} into $R_h(v,v') =
\mathbb{E}[h_m(T_s;v)\,h_m^*(T_s;v')]$ and using
$\mathbb{E}[\alpha_l\alpha_{l'}^*]=\delta_{ll'}^{\mr{K}}$
(Kronecker delta, superscripted to distinguish it from the velocity
spacing $\delta$), the double sum
over $(l,l')$ collapses to a single sum and the $p_{0,m}$-dependent
phases cancel exactly:
\begin{equation}
  R_h(v,v')
  = \frac{1}{L}\sum_{l=1}^{L}
     e^{j\frac{2\pi}{\lambda}(v-v')T_s\cos\theta_l}.
\end{equation}
For $\theta_l\sim\mathcal{U}(0,\pi)$ as $L\to\infty$, substituting
$u=\cos\theta$:
\begin{equation}
  R_h(v,v')
  = \int_{-1}^{1}\frac{1}{\pi}
    \frac{e^{j\frac{2\pi}{\lambda}(v-v')T_s u}}{\sqrt{1-u^2}}\,du
  = J_0\!\left(\frac{2\pi(v-v')T_s}{\lambda}\right). \quad
\end{equation}

\section{Derivation of the Average-Information Bound~\eqref{eq:CRLB}}
\label{app:crlb}

The $L$ distinguishable paths are assumed resolved into orthogonal delay bins
with independent noise. We first obtain an oracle CRLB for one resolved path,
conditioned on $\alpha_l$ and on the receiver-available AoD $\theta_l$; we then
sum the pathwise effective information, average it over the path gains and
isotropic AoDs, and invert only at the final step. The result is therefore an
average-information CRB, not the expectation of the conditional CRLB.

For the $l$-th path the unknown parameter vector is
$\bm{\xi}_l=[v_n,\phi_l]^T$, where $\phi_l=\angle(\beta_ls)$ is a nuisance
phase and $\theta_l$ is conditioned upon. The noiseless observation is
$\mb{u}_l=A_le^{j\phi_l}\mb{a}(f_{{\rm D},l,n})$ with $A_l^2=P|\beta_l|^2$ the
path power after averaging over the reliably detected constellation. Writing
$\mb{D}=\diag(0,\ldots,N-1)$, $c_l\triangleq T_s\cos\theta_l/\lambda$ and
$\kappa_l\triangleq2\pi c_l/N$, and using
$\partial f_{{\rm D},l,n}/\partial v_n=\cos\theta_l/\lambda$,
\begin{equation}
  \frac{\partial\mb{u}_l}{\partial v_n}=j\kappa_lA_le^{j\phi_l}\mb{D}\mb{a},
  \qquad
  \frac{\partial\mb{u}_l}{\partial\phi_l}=j\mb{u}_l .
  \label{eq:du}
\end{equation}
For $\mb{y}_l\sim\CN(\mb{u}_l,N_0\mb{I})$ with parameter-independent
covariance, the proper-complex FIM is
$[\mb{J}_l]_{ij}=\frac{2}{N_0}\mathrm{Re}\{(\partial\mb{u}_l/\partial
[\bm{\xi}_l]_i)^H(\partial\mb{u}_l/\partial[\bm{\xi}_l]_j)\}$. Since
$\mb{a}^H\mb{a}=N$, $\mb{a}^H\mb{D}\mb{a}=N(N-1)/2$ and
$\|\mb{D}\mb{a}\|^2=N(N-1)(2N-1)/6$, substituting~\eqref{eq:du} gives
\begin{equation}
  \mb{J}_l=\frac{2A_l^2}{N_0}
  \begin{bmatrix}
    \dfrac{(2\pi c_l)^2(N-1)(2N-1)}{6N} & \pi c_l(N-1)\\[8pt]
    \pi c_l(N-1) & N
  \end{bmatrix},
  \label{eq:FIM_J}
\end{equation}
with $\det(\mb{J}_l)=(2A_l^2/N_0)^2(2\pi c_l)^2(N^2-1)/12$. Eliminating the
nuisance phase by the Schur complement yields the scalar effective information
about $v_n$,
\begin{equation}
\begin{aligned}
  \mc{J}_{{\rm eff},l}
  &\triangleq[\mb{J}_l]_{11}-[\mb{J}_l]_{12}^2/[\mb{J}_l]_{22}\\
  &=\frac{A_l^2}{N_0}\frac{(2\pi T_s)^2(N^2-1)}{6\lambda^2N}\cos^2\theta_l
   =\frac{1}{\mr{CRLB}_l(v_n\mid\theta_l,\alpha_l)},
\end{aligned}
  \label{eq:Jeff_path}
\end{equation}
which is the conditional bound~\eqref{eq:crlb_cond2}.

It remains to average. With the beamformed spatial response
$g_{\mb{w}}(\theta)\triangleq\sum_mw_me^{j2\pi p_{0,m}\cos\theta/\lambda}$ we
have $\beta_l=\alpha_lg_{\mb{w}}(\theta_l)/\sqrt{L}$, so with
$a_{mm'}\triangleq2\pi(p_{0,m}-p_{0,m'})/\lambda$ and $\E[|\alpha_l|^2]=1$,
\begin{equation}
  \E_{\alpha_l,\theta_l}[|\beta_l|^2]
  =\frac{1}{L}\sum_{m,m'}w_mw_{m'}^*J_0(a_{mm'}).
  \label{eq:Ebeta2}
\end{equation}
When all nonzero antenna separations are large in wavelengths the off-diagonal
Bessel terms are negligible, and $\|\mb{w}\|^2=1$ gives
$\E[|\beta_l|^2]\approx1/L$; hence $\mr{SNR}\approx P/N_0$ and the mean SNR of
one path is $\mr{SNR}/L$. This corrects the path-power scaling without assuming
that $\E_{\alpha_l}[|\beta_l|^2\mid\theta_l]$ is constant in $\theta_l$.

The angular weighting follows from the integral representation of $J_0$: for
$\theta\sim\mc{U}(0,\pi)$,
\begin{equation}
  H(a)\triangleq\E_\theta\!\left[\cos^2\theta\,e^{ja\cos\theta}\right]
  =-J_0''(a)=J_0(a)-\frac{J_1(a)}{a},
\end{equation}
using the order-zero Bessel equation $J_0''+J_0'/a+J_0=0$ and $J_0'=-J_1$, with
$H(0)=1/2$ from $J_0(a)=1-a^2/4+O(a^4)$ and $J_1(a)=a/2+O(a^3)$. The same
expansion applied to $\E[|\beta_l|^2\cos^2\theta_l]$ replaces $J_0(a_{mm'})$ by
$H(a_{mm'})$ in~\eqref{eq:Ebeta2}, so that
\[
  \eta_{\mb{w}}
  =\frac{H(0)\|\mb{w}\|^2+\sum_{m\neq m'}w_mw_{m'}^*H(a_{mm'})}
        {J_0(0)\|\mb{w}\|^2+\sum_{m\neq m'}w_mw_{m'}^*J_0(a_{mm'})}
  \approx\frac{1/2}{1}=\frac{1}{2},
\]
the off-diagonal terms vanishing because $J_0(a)=O(|a|^{-1/2})$ and
$J_1(a)/a=O(|a|^{-3/2})$ for $|a|\gg1$.

Because the resolved paths carry independent noise their effective information
contributions add, so with $A_l^2=P|\beta_l|^2$,
\begin{equation}
  \overline{\mc{J}}_{\rm eff}
  =\E_{\bm{\alpha},\bm{\theta}}\!\left[\sum_{l=1}^{L}\mc{J}_{{\rm eff},l}\right]
  =\frac{(2\pi T_s)^2(N^2-1)}{6\lambda^2N}\,\mr{SNR}\,\eta_{\mb{w}},
\end{equation}
and inverting establishes~\eqref{eq:CRLB}. This must not be confused with
$\E[\mr{CRLB}_l]$: under Rayleigh fading $\E[1/|\alpha_l|^2]$ diverges, and for
a uniform AoD so does $\E[1/\cos^2\theta_l]$, so the literal mean of the
conditional bound is not finite, whereas the inverse expected-information bound
is a well-defined statistical design surrogate.

\bibliographystyle{unsrt}

\end{document}